\documentclass[11pt,reqno]{article}

\usepackage{amssymb,amscd,amsmath,amsthm,color,amsfonts,graphicx}
\usepackage{lineno}

\newtheorem{thm}{Theorem}[section]
\newtheorem{prop}[thm]{Proposition}
\newtheorem{lemma}[thm]{Lemma}
\newtheorem{cor}[thm]{Corollary}
\newtheorem{defn}[thm]{Definition}

\newtheorem{example}[thm]{Example}

\numberwithin{equation}{section}

\def\be{ \begin{linenomath} \begin{eqnarray}}
\def\ee{\end{eqnarray} \end{linenomath}}
\def\bee{  \begin{linenomath} \begin{eqnarray*}}
\def\eee{\end{eqnarray*}  \end{linenomath}}
 \def\pmx{\begin{pmatrix}}
 \def\emx{\end{pmatrix}}
 \def\bsq{\begin{subequations}}
\def\esq{\end{subequations}}

 \def\bst{\begin{subtheorem}}
\def\est{\end{subtheorem}}

\newcommand{\norm}[1]{ \| #1  \|}

\def\leqc{\preccurlyeq}

        \def\tr{\hbox{\rm Tr} \, }

     \def\half{{\textstyle \frac{1}{2}}}
     \def\nn{\nonumber}

\def\ds{\displaystyle}
 
\def\bra{\langle}
\def\ket{\rangle}
\def\kb{ \ket \bra }

\def\rt2{ \frac{1}{\sqrt{2}} }

\def\raw{\rightarrow}

\def\wh{\widehat}
           
\newcommand{\proj}[1]{ | #1 \kb  #1|}

\def\Raw{\Rightarrow}    

\def\ofxinv{ \big(x^{-1} \big) }
  \def\id{{\cal I}}
\def\ot{\otimes}

\def\bN{\mathbf{N}}
\def\bM{\mathbb{M}}
\def\bH{\mathbb{H}}
\def\bP{\mathbb{P}}
\def\bC{\mathbf{C}}
\def\bR{\mathbf{R}}
\def\cK{\mathcal{K}}
\def\cF{\mathcal{F}}
\def\<{\langle}
\def\>{\rangle}
\def\HS{\mathrm{HS}}
\def\Tr{\mathrm{Tr}\,}
\def\cD{\mathcal{D}}
\def\diag{\mathrm{diag}}
\def\ext{\mathrm{ext}}
\def\ffi{\varphi}
\def\WYD{\mathrm{WYD}}
\def\H{\mathrm{H}}
\def\B{\mathrm{B}}
\def\PD{\mathrm{PD}}
\def\Im{\mathrm{Im}\,}
\def\Re{\mathrm{Re}\,}
\def\St{\mathrm{St}}

\allowdisplaybreaks

\title{Families of completely positive maps\\ 
associated with monotone metrics}

\author{Fumio Hiai$^{1,}$\footnote{E-mail: fumio.hiai@gmail.com}\ ,
Hideki Kosaki$^{2,}$\footnote{E-mail: kosaki@math.kyushu-u.ac.jp}\ ,
D\'enes Petz$^{3,}$\footnote{E-mail: petz@math.bme.hu}\ \ and
Mary Beth Ruskai$^{4,}$\footnote{E-mail: mbruskai@gmail.com}}

\date{\today }

\begin{document}  

\maketitle

\begin{center}
$^1$\,Graduate School of Information Sciences, Tohoku University, \\
Aoba-ku, Sendai 980-8579, Japan
\end{center}
\begin{center}
$^2$\,Graduate School of Mathematics, Kyushu University, \\
Nishi-ku, Fukuoka, 819-0395, Japan
\end{center}
\begin{center}
$^3$\,Alfr\'ed R\'enyi Institute of Mathematics, H-1364 Budapest, POB 127, Hungary
\end{center}
\begin{center}
$^4$\,Tufts University, Medford, MA 02155, USA  \\ and \\
Institute for Quantum Computing, Waterloo, Ontario, Canada 
\end{center}

\begin{abstract}
\noindent
An operator convex function on $(0,\infty)$ which satisfies the symmetry condition
$k(x^{-1}) = x k(x) $ can be used to define a type
of non-commutative multiplication by a positive definite matrix (or its inverse)
using the primitive concepts of left and right multiplication and the functional
calculus.   The operators for the inverse can be used to define  quadratic 
forms associated with Riemannian metrics which contract under the action of
completely positive trace-preserving maps.   

We study the question of when these operators define maps which are also
completely positive (CP).  Although  $A \mapsto D^{-1/2} A D^{-1/2}$ is the only
case for which both the map and its inverse are CP, there are 
several well-known one  parameter families for which either the map or its
 inverse is CP.  
We present a complete analysis of the behavior of these families, as well as
the behavior of lines connecting an extreme point with the smallest
one and some results for geometric bridges between these points.  

Our primary tool is an order relation based on the concept of positive definite functions.
  Although some results can be obtained from known properties,
we also prove new results based on the positivity of the Fourier transforms
of certain functions.  Concrete computations of certain
Fourier transforms not only yield new examples of  
 positive definite functions, but also examples in the much stronger class of 
 infinitely divisible functions.

\bigskip

\noindent
{\it 2010 Mathematics Subject Classification: }
15A63, 15A60, 42A82, 46L60, 46L87, 26E05,

\medskip\noindent
{\it Key Words:}  monotone Riemannian metric; operator convex function; operator monotone function;
   completely positive map; positive definite kernel; infinite divisibility, quasi-entropy, geometric bridge

\end{abstract}


\tableofcontents

\section{Introduction} \label{sect:intro}

On a commutative algebra the operations of multiplication and ``division''
by elements of the positive cone take the positive cone into itself. However, this is not
the case for non-commutative algebras, on which these operations are not even uniquely
defined.

Various non-commutative versions of multiplication and division (i.e., multiplication by
the inverse) by elements of the cone of positive definite matrices or operators correspond
to maps on matrix algebras, and some of these maps play an important role in many
contexts.
For $D > 0 $, some naive definitions of multiplication by the inverse are
given by the maps $X \mapsto  D^{-1/2} X D^{-1/2}$, $X \mapsto D^{-1}X$ and
$X \mapsto X D^{-1}$. The first maps the cone of positive operators to itself, while
the other two do not even preserve self-adjointness. There are many other possible
generalizations of multiplication by $D^{-1}$.   In Section~\ref{sect:examps} we
 consider several different one parameter families of such maps.

Perhaps, the best known example is
\be  \label{BKM}
\Omega_D(X) = \int_0^\infty (D+tI)^{-1} X (D+tI)^{-1}\,dt \, ,
\ee
which gives a well-defined (and highly symmetric) notion of non-commutative multiplication
by $D^{-1}$.  Its inverse is well-known to be
\be  \label{BKMinv}
\Omega_D^{-1} (Y) = \int_0^1 D^t\,Y D^{1-t}\,dt
\ee
which is a form of non-commutative multiplication by $D$.  The quadratic form
$\tr A^* \Omega_D^{-1} (B)$ is known as the Bogoliubov or Kubo-Mori inner product.  
Although the inverse relationship between \eqref{BKMinv} and  \eqref{BKM}  is well-known,
and follows from more general results given later, 
 we include an explicit proof in Appendix~\ref{app:BKM}. 

We consider here a class of such maps which arise in quantum information theory,
in the context of what are known as monotone Riemannian metrics \cite{Pz2,LR}.

We study the question of which maps within this class have the property known as
completely positivity defined in Section~\ref{sect:CP}. The map in $ \eqref{BKM}$ has
this property, but its inverse \eqref{BKMinv} does not.
 
This question was motivated by an observation in \cite{LR,TKRWV} about bounds on the
contraction of monotone metrics under the action of completely positive trace-preserving
(CPT) maps which are also known as {\em quantum channels} in view of their important role
as noise models in quantum information theory.  These bounds are discussed briefly in
Section~\ref{sect:back} and Appendix~\ref {app:contractbd}. 
 
This paper is organized as follows.   Section~\ref{sect:prelim}  describes the various
concepts we need and introduces the notations we will use.  Section~\ref{sect:back}
contains a brief summary of the background and motivation behind this work.   In
Section~\ref{sect:posdef} we present some powerful tools we will use formulated
in terms of positive kernels, and a closely related partial order.
In Section~\ref{sect:examps} we present a large number of examples of one
parameter families of functions which provide a number of inequivalent 
classes of maps  used to define non-commutative multiplication by the inverse
of a positive matrix.   In most cases, we can also provide precise ranges 
for which these maps or their inverses are CP.    Section~\ref{sect:proofs}
proves new results about positive kernels based on Fourier 
transforms,  which are needed to prove some of our results.  
These results are of interest in their own right.
In Section~\ref{sect:exampfs} we complete the proofs of those
results stated in Section~\ref{sect:examps} which require the
results of Section~\ref{sect:proofs}.
  
There are also three appendices.   The first contains the detailed 
proof of the integral representation needed in Section~\ref{sect:basics}.
The second gives more details about the motivation in terms of 
contraction under CPT maps of the Riemannian metrics associated
with our maps.  Finally, for the benefit of non-experts, we present some
very pedestrian arguments which clarify some well-known results that
are often glossed over.

\bigskip 
\section{Preliminaries} \label{sect:prelim} 

\subsection{Basics} \label{sect:basics}

For each $d\in\bN$ we write $\bM_d$, $\bH_d$, $\bP_d$ and  $\overline{\bP}_d$
for the sets of $d\times d$ complex, Hermitian, positive definite and
positive semi-definite matrices, respectively. Functions of matrices in $\bH_d$ can be
defined by using the spectral theorem; this is sometimes called ``functional calculus"
(see, e.g., \cite[Section VII.1]{RS1}).

A real function $f$ on $(0,\infty)$ is said to be {\it operator monotone} (or operator
monotone increasing) if $A\ge B$ implies $f(A)\ge f(B)$ for every $A,B\in\bP_d$ with any
$d\in\bN$, {\it operator monotone decreasing} if $-f$ is operator monotone. A real function
$k$ on $(0,\infty)$ is said to be {\it operator convex} if 
$$
k(\lambda A+(1-\lambda)B)\le\lambda k(A)+(1-\lambda)k(B)
$$
for all $A,B\in\bP_d$ with any $d\in\bN$ and all $\lambda\in(0,1)$, and {\it operator
concave} if $-k$ is operator convex.  The theory of operator monotone and operator convex functions
was initiated by L\"owner \cite{Lw} and Kraus \cite{Kr}, respectively.   It is well-known
\cite[Section V.4]{Bh} (also \cite{Ando,Do,Hi}) that operator monotone (also operator
convex) functions on $(0,\infty)$ have an analytic continuation into the upper half-plane
of $\bC$. Moreover, a necessary and sufficient condition for a function on $(0,\infty)$ to
be operator monotone increasing (resp., decreasing) is that it has the
{\em ``Pick'' mapping property} that the analytic continuation maps the upper
half-plane {\em into} the upper (resp., lower) half-plane.  The integral
representation theory for Pick functions and operator monotone functions are also
well-known, see, e.g., \cite[Section 59]{AG} or \cite[Section V.4]{Bh} and
\cite{Ando,Do,Hi}.

\begin{defn}\rm
Let  $\cK$  denote the class of functions $k:(0,\infty)\to(0,\infty)$ which are operator
convex and satisfy the symmetry condition  $xk(x)=k\ofxinv$ and the normalization
condition $k(1) = 1$.  
\end{defn}

There are a number of equivalent characterizations of the class $\cK$ which are given in
Theorem~\ref{thm:kequiv} below. Its proof uses an integral representation, which is
important in its own right and presented in Theorem~\ref{thm:intrep}. We also observe that 
although $k(x)  \in\cK$ may diverge as $x \searrow 0$, it cannot diverge more
rapidly than $x^{-1}$. This was proved in \cite{HS}. For completeness, we include
its proof as well as the proof of Theorem 2.3 in Appendix~\ref{app:intrep}.

\begin{prop} \label{prop:lim0}
Let $k:(0,\infty)\to(0,\infty)$ be operator convex. Then $\lim_{x \raw 0 } x k(x)$
exists and is finite. When $xk(x)=k\ofxinv $, $\lim_{x\raw \infty} k(x) $ also exists and
is finite.
\end{prop}
 
\begin{thm} \label{thm:intrep}
For every $k \in \cK$, there exists a unique probability measure $m$ on
$[0,1]$ such that for $x \in (0,\infty)$
\begin{align*}
k(x)
& = \int_{[0,1]} \frac {1+x}{ (x+\nu)(1+\nu x)}\cdot \frac{(1+\nu)^2 }{2}\,dm(\nu) \\
& = \int_{[0,1]} \bigg( \frac{1}{x + \nu}  +  \frac{1}{1 +  x \nu} \bigg)
\frac{(1+ \nu)}{2}\,dm(\nu).
\end{align*}
\end{thm}

\begin{thm} \label{thm:kequiv}
For each function $k:(0,\infty)\to(0,\infty)$ the following are equivalent{\rm:}
\begin{itemize}
\item[\rm(a)] $k$ is operator convex and $xk(x)=k\ofxinv${\rm;}
\item[\rm(b)] $k$ is operator monotone decreasing and $xk(x)=k\ofxinv${\rm;}
\item[\rm(c)] $f(x)\equiv 1/k(x)$ is operator concave and $f(x)=xf\ofxinv${\rm;}
\item[\rm(d)] $f(x)\equiv 1/k(x)$ is operator monotone and $f(x)=xf\ofxinv$.
\end{itemize}
\end{thm}

\begin{proof}
The equivalence of $xk(x)=k\ofxinv$ and $f(x)=xf\ofxinv$ is easily checked. The
equivalence (c) $\Leftrightarrow$ (d) for positive functions on $(0,\infty)$ is well-known
(see, e.g., \cite[Theorem V.2.5]{Bh}), and (b) $\Leftrightarrow$ (d) follows from the fact
that $x\mapsto1/x$ is operator monotone decreasing on $(0,\infty)$.  The implication
(a) $\Raw$ (b) follows immediately from Theorem~\ref{thm:intrep} and the well-known fact
that the map  $x \mapsto 1/(\alpha x + \beta)$ is operator monotone decreasing on
$(0, \infty)$ for any fixed $\alpha, \beta \geq  0$. We finally show that (b) $ \Raw $ (a).
Assume (b); then $1/k(x)$ is operator monotone and so operator concave on $(0,\infty)$.
This implies (a) since $x^{-1}$ is operator monotone decreasing and operator convex.
\end{proof}

As shown in the above proof, the implication (b) $\Raw $ (a) and the equivalence of
(b)--(d) hold true without the symmetry assumption $x k(x) = k\ofxinv $. However, the
reverse implication (a) $\Raw $ (b) only holds under this additional assumption and appears
to be new.

The next result is easy to verify, but stated explicitly for completeness.
\begin{prop} \label{prop:hat}
The map $k(x)  \mapsto \wh{k}(x) \equiv 1/k\big( x^{-1} \big)$ is a bijection on $\cK$ and
the map $f(x) \mapsto \wh{f}(x)  \equiv 1/f \big( x^{-1} \big)$  induces the same bijection
with $f(x)=1/k(x)$ and $\wh{f}(x) = 1/\wh{k}(x)$.
\end{prop}
 
\subsection{The multiplication map and its inverse $\Omega_D^k$} \label{sect:maps}

For each $D\in\bP_d$ we write $L_D$ and $R_D$ for the left and the right multiplication
operators, respectively, i.e., $L_DX\equiv DX$ and $R_DX\equiv XD$ for $X\in\bM_d$. Note
that $L_D$ and $R_D$ are commuting positive invertible operators on $\bM_d$ considered as
a Hilbert space equipped with the Hilbert-Schmidt inner product
$\<X,Y\>_\HS\equiv \Tr X^*Y$, where $\Tr$ denotes the usual trace functional on $\bM_d$.  
The operator $L_A R_B^{-1} $ was used by Araki \cite{Ak} to define the relative entropy
of positive operators $A,B$ in far more general situations than matrix algebras, and is
often called the {\em relative modular operator}.

For a fixed function $k\in\cK$ we define, for any $D\in\bP_d$, the linear map
$\Omega_D^k:\bM_d\to\bM_d$ by
\begin{equation}\label{omegadef}
\Omega_D^k(X)\equiv R_D^{-1}k\left(L_DR_D^{-1}\right)X
=L_D^{-1}k\left(R_DL_D^{-1}\right)X,\qquad X\in\bM_d.
\end{equation}
Since both of the commuting operators $L_D$ and $R_D$ are positive with respect to the
Hilbert-Schmidt inner product, it is clear that $\Omega_D^k$ is also positive (in the same
sense). Each map $\Omega_D^k$ can be considered as a non-commutative generalization of
multiplication by $D^{-1}$; indeed, if $DX=XD$ then $\Omega_D^k(X)=D^{-1}X$
independently of $k\in\cK$. The equivalence of the two expressions in \eqref{omegadef}
follows from the symmetry condition $xk(x)=k\ofxinv$ since    
\begin{linenomath} \begin{align*}
R_D^{-1}k\left(L_DR_D^{-1}\right)&=L_D^{-1}L_DR_D^{-1}k\left(L_DR_D^{-1}\right) \\
&=L_D^{-1}k\left(\left(L_DR_D^{-1}\right)^{-1}\right)=L_D^{-1}k\left(R_DL_D^{-1}\right).
\end{align*} \end{linenomath}

To better understand the action of $\Omega_D^k$, we consider the two-variable
function\linebreak $ \phi^k(x,y)\equiv  (1/y) \, k(x/y) $ for $x,y>0$ and observe
that $\Omega_D^k = \phi^k(L_D,R_D)$. When $D$ is a diagonal matrix with
eigenvalues $\lambda_j$, it is an easy consequence of the functional calculus that the
action of $\Omega_D^k$ on the matrix with entries $x_{ij}$ is
$$
x_{ij} \longmapsto \phi^k(\lambda_j, \lambda_k) x_{ij}
= \frac{1}{ \lambda_j} k\biggl( \frac{\lambda_i}{\lambda_j}\biggr) x_{ij}
$$
which is the Schur (or Hadamard or pointwise) product  $A \circ X$ with the matrix $A$
with entries $a_{ij} = \phi^k(\lambda_i, \lambda_j)$.  More generally, let
$U$ be a unitary which diagonalizes $D$ so that
$$
D = U\diag(\lambda_1,\dots,\lambda_d)U^*.
$$ 
Then  
\be \label{schurgen}
\Omega_D^k(X)=U \Big( [\phi^k(\lambda_i, \lambda_j )]\circ [U^*XU] \Big)U^*.
\ee 

Since $L_D^{-1}=L_{D^{-1}}$ and $R_D^{-1}=R_{D^{-1}}$, it might be tempting to think that
$\Omega_D^{-1}(X)=\Omega_{D^{-1}}(X)$.  However, this is easily seen to be false by
considering specific examples (including \eqref{BKM} and \eqref{BKMinv}).  
Instead we have for any $D \in \bP_d$,
\begin{equation} \label{inv}
J_D^f \equiv (\Omega_D^k)^{-1} =R_D f \big( L_DR_D^{-1} \big)  
=R_D\wh k(L_D^{-1}R_D)=\Omega_{D^{-1}}^{\wh k}
\end{equation}
with $f(x) = 1/k(x) =\wh{k}\ofxinv$.  To see this, observe that it follows from
\eqref{schurgen} that
$(\Omega_D^k)^{-1}=(1/\phi^k)(L_D,R_D)$. From the relation 
$$
\frac{1}{\phi^k}(x,y)=\frac{y}{k(x/y)}=  y f(x/y) =  x f(y/x),
$$
the functional calculus implies \eqref{inv}. 

\subsection{Complete positivity of $\Omega_D^k $} \label{sect:CP}

A linear map $\Phi : \bM_d \mapsto \bM_d$ is called {\em positive} if it is
positivity-preserving in the sense that $A >0$ implies $\Phi(A) \geq 0$, i.e.,
$\Phi(\bP_d) \subseteq \overline{\bP}_d$.  A linear map
$\Phi : \bM_d \mapsto \bM_d$ is called  {\em completely positive} (CP) if $\Phi \ot \id_n$
is positive on $\bM_d \ot \bM_n$ for all $n \in \bN$ with $\id_n$ the identity map on
$\bM_n$. The notion of complete positivity, introduced by Stinespring \cite{St} and
discussed in, e.g.,  \cite[Chapter~6]{Paul}
 plays an important role in quantum information theory.  (See, e.g,  \cite{NC,pxbk}.)

The recognition in \eqref{schurgen} that $\Omega_D^k$ can be implemented as a Schur
product yields a simply stated condition to test that it is CP. In general, complete
positivity of a map on $\bM_d$ is a much stronger condition than positivity. However, for
Schur products, it is well-known (see, e.g., \cite[Theorem~3.7]{Paul}) that both positivity
conditions for the map $\Phi_A(X)=A \circ X$ ($A,X \in \bM_d$) are equivalent to
positivity of $A$. Indeed, the map $\Phi_A \otimes \id_n$ on
$\bM_n(\bM_d) \cong \bM_d \otimes \bM_n$ can be realized as Schur multiplication with
$A \otimes J_n$, where $J_n$ is the $n \times n$ matrix with all entries $1$. Therefore,
in our setting, the requirement that the map $\Omega_D^k$  
is CP is equivalent
to the weaker positivity requirement, as we explicitly state for completeness in the
following:
\begin{thm} \label{thm:CP1}
The following conditions for $k\in\cK$ are equivalent{\rm:}
\begin{itemize}
\item[\rm(a)] $\Omega_D^k:\bM_d\to\bM_d$ is CP for every $D\in\bP_d$ with any $d\in\bN${\rm;}
\item[\rm(b)] $\Omega_D^k:\bM_d\to\bM_d$ is positive for every $D\in\bP_d$ with any
$d\in\bN${\rm;}
\item[\rm(c)] the $d \times d$ matrix
\begin{equation} \label{Adef}
A = \biggl[ \frac{1}{w_j} k\biggl( \frac{w_i}{w_j}\biggr)\biggr]_{1\le i,j\le d}    
\end{equation}
is positive semi-definite for every $w_1,\dots,w_d>0$ with any $d\in\bN$.
\end{itemize}
\end{thm}

The next result allows us to replace the matrix $A$ in part (c) of Theorem~\ref{thm:CP1}
by some closely related matrices   for which the positivity condition may be more easily
checked in some situations.

\begin{prop} \label{prop:Xalt}
The matrix $A$ in \eqref{Adef} is positive if and only if one {\rm(}and hence both{\rm)\,}
of the following matrices are positive{\rm:}
\be \label{Aalt}
\biggl[w_ik\biggl( \frac{w_i}{w_j} \biggr)\biggr]_{1\le i,j\le d},  \qquad
\biggl[ \sqrt{ \frac{w_i}{w_j}} k\biggl(\frac{w_i}{w_j}\biggr)\biggr]_{1\le i,j\le d}.     
\ee  
\end{prop} 

\begin{proof}
Let $W$ be the diagonal matrix with entries $w_i \delta_{ij}$. Then the matrices above
correspond to $W^*AW$ and $(W^*)^{1/2} A W^{1/2}$, respectively.
\end{proof}

\subsection{Background and motivation}  \label{sect:back} 

For each $k \in \cK$, the map $\Omega_D^k$ can be used to define a quadratic form
\be \label{riemk}
\Gamma_D^k(X,Y) \equiv \bra X, \Omega_D^k(Y) \ket_\HS = \tr X^* \Omega_D^k(Y) 
\ee
which can be interpreted as a metric on the Riemannian manifold
$\cD_d \equiv \{ D \in \bP_d : \tr D = 1\}$ of invertible density matrices in $\bM_d$. 
Here, the  matrices in $\bH_d $ with trace zero form the tangent space, denoted by
$\bH_d^0$, of $\cD_d$ at each foot point $D$.  This metric is {\em monotone} in the sense
that for any completely positive and trace-preserving map $\Phi:\bM_d\raw\bM_m$
($d,m\in\bN$),
\be \label{contract}
\Gamma_{\Phi(D)}^k\big(\Phi(X), \Phi(Y) \big) \leq \Gamma_D^k(X,Y), \qquad
D \in \cD_d,\ X,Y\in\bH_d^0.
\ee
The theory of monotone Riemannian metrics was largely developed
 by Petz \cite{Pz2} after Morozova and Chentsov \cite{MC}
introduced the concept. It was shown  in \cite{Pz2} that each $k \in\cK$
defines a family of monotone metrics of the form \eqref{riemk} with $D \in \cD_d$ for all
$d \in  \bN$, and that any Riemannian metric on $\cD_d$, $d\in\bN$, which satisfies the
contraction condition \eqref{contract} must be of the form \eqref{riemk} for some
$k \in \cK$.  (See also \cite{Kum}.)

In  \cite{Pz2}, the operator $J_D^f$ defined in \eqref{inv} was used to define
monotone metrics in the equivalent form as
$$
\Gamma_D^k(X,Y) = \bra X, (J_D^f)^{-1}(Y) \ket_\HS = \tr X^* (J_D^f)^{-1} (Y).
$$
It might seem more natural to work with $ J_D^f $ which is a non-commutative version of
multiplication by $D$ rather than using its inverse $\Omega_D^k$ (introduced in \cite{LR}).    
However, in this paper we use $\Omega_D^k$ instead of $J_D^f$ since it
avoids the need to take inverses to define our target maps.

In \cite{TKRWV}, monotone metrics of the form  $\Gamma_Q^k(P-Q,P-Q)$ with  $P, Q \in \cD_d$
played an important role in the study of mixing times of Markov processes. It was observed
in \cite[Section III.B]{TKRWV} and \cite[Section IV.C]{LR} that when both $\Omega_D^k$
and its inverse $(\Omega_D^k)^{-1}$ are positivity-preserving, one can obtain a useful
upper bound on the contraction of Riemannian metrics, which is described in more detail
in Appendix~\ref{app:contractbd}. In \cite{TKRWV} this bound was used in the case
$k(x) = x^{-1/2}$ for which both $\Omega_D^k (A) =  D^{-1/2} A D^{-1/2}$
and the inverse $ D^{1/2} A D^{1/2} $ clearly map $\bP_d$ into itself. (In fact, for every
$D \in \bP_d$  they are bijections on $\bP_d$.) Theorem~\ref{thm:unique} below implies that
$k(x) = x^{-1/2}$ is the only function in $\cK$ with this property.

The study of {\em quasi-entropies} was also initiated by Petz in
\cite{Pz0,Pz1,OP}, which can be defined from any operator convex function $g$ on $(0,\infty)$
with $g(1) = 0$ as
\be \label{quasi}
H_g(A,B,K) \equiv \bra K,g\big(L_A R_B^{-1}\big)R_BK\ket_\HS
= \tr \sqrt{B} K^* g \big(L_A R_B^{-1} \big)(K \sqrt{B} )
\ee
for $A,B \in \bP_d, K \in \bM_d $.  It was later observed in \cite{Pz3,PH} that
for any $D \in \cD_d$ the Hessian
\bee
- \frac{ \partial^2}{\partial a \partial b }\,H_g(D + a X, D + b Y, I)\Big|_{a=b=0},
\qquad X, Y \in \bH_d^0,
\eee
can be associated with a monotone Riemannian metric in some important examples.
This was   proved  for more general $g$ in
\cite[Theorem II.8]{LR}, where it was also noticed that  
$g(x) = (1-x)^2 k(x) $ with $k \in \cK$ is an operator convex function on $(0,\infty)$
with $g(1) = 0 $.   Moreover,  the symmetry condition $x k(x) = k \ofxinv $
implies that $ x g \ofxinv = g(x) $ and the quasi-entropy with $K=I$ has the symmetry
property
\be \label{sym}
H_g(A,B, I) =  H_g(B,A, I).
\ee
and that every quasi-entropy with this symmetry property comes from a $k \in \cK$.
The quantity $H_g( A,B, I) $ is often called an {\em $f$-divergence}. See \cite{HMPB} for
a thorough discussion of $f$-divergences (without the additional symmetry condition).  

When $k(x) = 4/(1 + \sqrt{x})^2$ so that $g(x) = 4(1-\sqrt{x})^2$, the function
$H_g(A,A,K)$ is the {\em Wigner-Yanase skew information} \cite{WY}, which Dyson suggested
extending to the case including the parameter $p \in (0,1)$, that is equivalent to using
$g(x) = 4(1 - x^p)(1 - x^{1-p})$. This led to Lieb's seminal work  on concave trace
functions \cite{Lb}, in which he showed that $(A,B) \mapsto \tr K^* A^p K B^{1-p} $ is
jointly  concave in $A, B \in \bP_d$ when $p \in (0,1)$. It is
implicit\footnote{Ando found an alternate proof of Lieb's concavity results and also
showed convexity for $p \in (1,2)$. Both Lieb and Ando ignored the linear term
$\tr K^* A K$ in the skew information, since it was irrelevant to convexity.} in Ando's
paper \cite{Ando2} that the quasi-entropy  $H_g(A,B,K)$ can be extended to
$g(x) = (1 - x^p)(1 - x^{1-p})/p(1-p)$ with $p \in [-1,2]$. Hasegawa \cite{Has} seems to
have been the first to use well-known properties of monotone and convex operator functions
to explicitly recognize that replacing $4$ by $1/p(1-p)$ allows one to extend the
quasi-entropy\footnote{Hasegawa actually used the asymmetric $g(x) = (1 - x^p)/p(1-p)$.
However, it follows from Eq.~(37) in \cite{LR} that this yields the same $k(x)$ given by
\eqref{WYD} as the symmetric version above.} for the WYD skew information and the
associated Riemannian metric to the maximal range  $p \in [-1,2]$ (with $p = 0,1$ defined
as limits{\footnote{Lindblad \cite{Lind} was the first to observe that one could recover
joint convexity of the usual relative entropy by taking $\lim_{p \raw 1} $ in Lieb's
result.}).  See also \cite{JR} where equality conditions were given for  the convexity of
$H_g(A,B,K) $ and some other inequalities for the extended WYD family. 
    
In this paper, we make use of tools developed  by Hiai and Kosaki \cite{HK1,HK2}
in study of means of operators.\footnote{This work was motivated by inequalities for
unitarily invariant norms.  The term {\em mean} used there does not, in general, yield
the mean of a pair of operators in the sense of Kubo and Ando \cite{KA}.} Motivated by
this work, whenever $k\in\cK$, we define
\begin{align}
M^k(x,y) & \equiv {y \over k(x/y)}, \qquad x,y>0, \label{HKmean}\\
M^k(A,B) & \equiv R_B \Bigl(k\bigl(L_A R_B^{-1} \bigr)\Bigr)^{-1}, \quad\ A,B \in \bP_d.
\label{HKopmean}
\end{align}
From \eqref{omegadef} and \eqref{inv} we have in particular
\be \label{M-Omega}
M^k(D,D) = \bigl(\Omega_D^k\bigr)^{-1} = \Omega_{D^{-1}}^{\wh{k}} = J_D^f.
\ee
The function $M^k(x,y)$ is called a {\it symmetric homogeneous mean} for positive scalars,
i.e., $M=M^k:(0,\infty)\times(0,\infty)\to(0,\infty)$ is a continuous function such that
\begin{itemize}
\item[(1)] $M(x,y)=M(y,x)$,
\item[(2)] $M(tx,ty)=tM(x,y)$ for $t>0$,
\item[(3)] $M(x,y)$ is non-decreasing in $x,y$,
\item[(4)] $\min\{x,y\}\le M(x,y)\le \max\{x,y\}$.
\end{itemize}
With $f = 1/ \wh{k}$, definition \eqref{HKmean} is equivalent to
$M^k(x,y) = y \,f(x/y)$ which was used in \cite{HK1,HK2} under the weaker
condition that $f$ is non-decreasing in the numerical sense. It follows from
Proposition~\ref{prop:hat} that as $k$ runs through $\cK$ both conventions generate the
same set of operators of the form \eqref{HKopmean}.

\subsection{The convex sets $\cK$ and $\cK^+$}  \label{sect:Kext}

Recall that $\cK$ denotes the set of functions $k:(0,\infty)\to(0,\infty)$ satisfying any
of the equivalent conditions of Theorem \ref{thm:kequiv}  and $k(1)= 1$.  With
$\wh{k}(x) = 1 / k \ofxinv $ given in Proposition \ref{prop:hat},
$k \mapsto \wh{k}$ is a bijective transformation on $\cK$.

For $\nu\in[0,1]$ let us set
\be \label{extpt}
k_\nu^\ext \equiv \frac{(1+\nu)^2 }{ 2} \cdot \frac {1+x }{ (x+\nu)(1+\nu x)}
= \frac{(1+ \nu)}{2}  \bigg(  \frac{1}{x + \nu}  +  \frac{1}{1 +  x \nu} \bigg).
\ee
For any fixed $x\in(0,\infty)$}, by computing the derivative of $k_\nu(x)$ in $\nu$ one
can easily verify that $k_\nu(x)$ is non-increasing in $\nu\in[0,1]$ so that
\be \label{extptord}
k_1^\ext(x)  = \frac{2}{1+x} \le k_\nu^\ext(x) \le
\frac{1+x}{2x} = k_0^\ext(x),\qquad \nu\in(0,1).
\ee
Since Theorem~\ref{thm:intrep} can be rewritten as
\be \label{intrepnu}
k(x) = \int_{[0,1] }  k_\nu^\ext(x)  ~ dm(\nu)
\ee
with $m$ a probability measure, one moreover has
\be \label{genpword}
\frac{2}{1+x}  \le  k(x) \le \frac{1+x}{2x}, \qquad k \in \cK.
\ee 
Thus, $\cK$ has the smallest element $k_1^\ext(x)=2/(1+x)$ and the largest element
$k_0^\ext(x)=(1+x)/2x$ in the pointwise order.

Now we may consider $\cK$ as a subset of the locally convex topological vector space
consisting of real functions on $(0,\infty)$ with the pointwise convergence topology.
Then it is obvious from \eqref{genpword} that $\cK$ is a convex and compact subset.
The uniqueness of the representing measure $m$ in Theorem~\ref{thm:intrep} implies that
$\cK$ is a Choquet simplex with the extreme points $k_\nu^\ext$ for $\nu\in[0,1]$ (that is
the reason for the notation $k_\nu^\ext$). Furthermore, since $\nu\mapsto k_\nu^\ext$ is
a homeomorphism from the interval $[0,1]$ into $\cK$, one sees that $\cK$ is a so-called
Bauer simplex (as in \cite{HP}). 

Motivated by the work on contraction bounds in \cite{LR, {TKRWV}} which is described in
Appendix~\ref{app:contractbd}, we define two subsets $\cK^+$ and $\cK^-$ of $\cK$ as
\begin{linenomath} \begin{align*}
\cK^+&\equiv \{k\in\cK:\mbox{$\Omega_D^k$ is CP for every $D\in\bP_d$, $d\in\bN$}\}, \\
\cK^-&\equiv \{k\in\cK:\mbox{$(\Omega_D^k)^{-1}$ is CP for every $D\in\bP_d$, $d\in\bN$}\}.
\end{align*} \end{linenomath}
 It follows from \eqref{inv} that 
\begin{equation}\label{relK^+K^-}
k\in\cK^+\Longleftrightarrow \wh{k} \in\cK^-\quad\mbox{where}\quad\wh{k}(x)=1/k\ofxinv.
\end{equation}

It follows from from Theorem \ref{thm:CP1} that $\cK^+$ and $\cK^-$ are closed under
pointwise convergence.  Although $\cK^+$ is convex,  $\cK^-$ is not convex 
(as shown in Example \ref{Example 4.4} below).  Since $\cK^+$ is a compact convex subset of $\cK$,
 it is the closed convex hull of its 
extreme points by the Krein-Milman theorem.  However, determining all the extreme points of
$\cK^+$ seems quite challenging.   Some non-trivial  ones are described in
Example \ref{Example 4.3}. 

In this paper, we have chosen to formulate most of our results in
terms of functions $k \in \cK$. Ê As is clear from Theorem~\ref{thm:kequiv}
we can also define the convex set of functions $\cF$ with $f = 1/k$ which
satisfy property (c) or (d). Ê Although our choice is
partly a matter of taste, in some situations,
one may be more convenient than the other.
We find it useful here to let
$$
\cF^{\pm} \equiv \{f \in \cF : \mbox{$(\Omega_D^{1/f})^{\pm 1} $ is CP for every
$D\in\bP_d$, $d\in\bN$}\},
$$
so that $\cK^{\pm} $ corresponds to $\cF^{\pm} $ by $k\leftrightarrow f=1/k$.
Since $1/k(x)=\wh{k}(x^{-1})=x\,\wh{k}(x)$, it is obvious by \eqref{relK^+K^-} that
$$
\cF=\{xk(x):k\in\cK\},\qquad
\cF^+=\{xk(x):k\in\cK^-\},\qquad
\cF^-=\{xk(x):k\in\cK^+\}.
$$
Hence $k\leftrightarrow xk(x)$ gives an affine correspondence between $\cK$ and $\cF$, by
which $\cK^+$ is isomorphic to $\cF^-$. Therefore, $\cF^-$ is also convex
and the extreme points of $\cF$ are  
\be
f_\nu^\ext(x) = x\, k_\nu^\ext(x)
=Ê\frac{(1+\nu)^2 }{ 2} \cdot \frac {x \, ( 1+x) }{ (x+\nu)(1+\nu x)}.
\ee


\section{Positive kernels and induced order}  \label{sect:posdef}

\subsection{Basic definitions}

In principle, the condition of Theorem~\ref{thm:CP1}\,(c) gives a simple criterion for
complete positivity. But in practice, it is not easy to verify that either the matrix $A$
in \eqref{Adef} or one of those in \eqref{Aalt} is positive semi-definite. Only a few
examples can be resolved using this criterion. However, there is another equivalent
condition based on the theory of functions which define positive kernels.

\begin{defn}\label{def:PD/ID} \rm \
 A continuous function $h: \bR \mapsto \bC$ is called
{\em positive definite} if $h(x-y)$ is a positive semi-definite kernel, i.e.,
$\bigl[h(t_i-t_j)\bigr]_{1\le i,j\le d}$ is positive semi-definite for any
$t_1,\dots,t_d\in\bR$ with any $d\in\bM$, or equivalently,
$$  
\iint \overline{\ffi(s)} h(s-t) \ffi(t)\,ds\,dt \geq 0,
\qquad  \ffi \in C_0^\infty(\bR),  
$$
where $C_0^\infty(\bR)$ denotes the smooth compactly supported functions on $\bR$.
Functions satisfying this condition are sometimes called  ``functions of positive type'' 
\cite[Section IX.2]{RS2} or ``positive in the sense of Bochner''.

Moreover, $h$ is called {\it infinitely divisible} if $h(t)^r$ is positive
definite for every $r>0$, or equivalently,  $h(t)^{1/n}$ is positive
definite for every $n\in\bN$.
\end{defn}

For convenience, some basic properties of positive definite functions stated here:
\begin{itemize}
\item[(a)] A positive definite function $h$ is uniformly bounded on $\bR$ as
$|f(t)| \le f(0)$.
\item[(b)] Bochner's theorem (see \cite[Theorem~IX.9]{RS2}, \cite[Section 60]{AG}) says
that $h$ is positive definite if and only if it is the Fourier transform of a finite
positive measure on $\bR$. Thus, positive definiteness of $h$ can be checked, in principle,
by testing positivity of its Fourier transform.

\item[(c)] The product of positive definite functions is positive definite. This
immediately follows from the well-known fact that the Fourier transform of the convolution
of two finite measures is the product of their Fourier transforms, or from the Schur
product theorem for positive semi-definite matrices.
\end{itemize}

In this paper we only consider positive definite functions on $\bR$ so that we shall omit
``on $\bR$" in the rest. Positive definite functions played an important role in the work
\cite{HK1,HK2} on means of operators, where a partial order was introduced. The following
definition is its adaptation to functions in $\cK$:

\begin{defn} \label{def:order}\rm
For $k_1,k_2\in\cK$ we write $k_1\leqc k_2$ if either of the following equivalent
conditions holds\,:
\begin{itemize}
\item[\rm(a)] the function $k_1(e^t)/k_2(e^t)$ is positive definite on $\bR$\,;
\item[\rm(b)] the matrix
$$
\ds{\biggl[{k_1(w_i/w_j)\over k_2(w_i/w_j)}\biggr]_{1\le i,j\le d} }
$$
is positive semi-definite for every $w_1,\dots,w_d>0$ with any $d\in\bN$.
\end{itemize}
\end{defn}

It is easily verified as in \cite{HK1,HK2} that $\leqc$ is really a partial order in $\cK$,
and $k_1\leqc k_2$ implies $k_1(x)\ \le k_2(x)\ $  on $(0,\infty)$, i.e., $ k_1 \le k_2 $ pointwise.   

The stronger condition  that $k_1(e^t)/k_2(e^t)$ is
infinitely divisible (following Definition~\ref{def:PD/ID}),  was studied in \cite{BK}.  
Results given there sometimes play a role 
in showing that the one-parameter families studied in Section~\ref{sect:classic} are
monotonic in the $\leqc$ order.  Moreover, infinite divisibility is important in the
discussion of geometric bridges in Sections~\ref{sect:geombrdg} and \ref{sect:geombrdgpf}. 
Some examples considered here require new results for specific functions which are obtained 
in Sections~\ref{sect:infdiv} and \ref{sect:infdiv2}.

The next useful lemma on positive definite functions will often be used in this paper.
See \cite[Appendix B]{Ko0} and \cite[Theorem 3.2]{BhPa} for the proof of (1). 
On the other hand, (2) was first proved in \cite[Theorem 5.1]{BhPa} while
the ``if part" was pointed out earlier in \cite{Zh}.

\begin{lemma} \label{lemm:hyperb} \
\begin{itemize}
\item[\rm(1)] The function ~$\sinh\alpha t/\sinh t$ is positive definite for
$\alpha\in(0,1)$.
\item[\rm(2)] For $\beta>-1$, the function ~$(\cosh t+\beta)^{-1}$ is positive definite if
and only if $\beta\le1$.
\end{itemize}
\end{lemma}

\subsection{Basic applications}

The next theorem gives a basic characterization of the class $\cK^+$. The equivalence of
(a)--(c) follows immediately from Theorem \ref{thm:CP1} with $A$ replaced by the second
matrix in \eqref{Aalt}. The equivalence of (b) and (d) is an adaptation of
\cite[Theorem 1.1]{HK1} via \eqref{M-Omega} in the present situation.

\begin{thm} \label{thm:CP2}
The following conditions for $k\in\cK$ are equivalent{\rm:}
\begin{itemize}
\item[\rm(a)]   $k \in \cK^+ $, i.e., $\Omega_D^k$ is CP for every $D\in\bP_d$ with
any $d\in\bN${\rm;}
\item[\rm(b)] $k\leqc x^{-1/2}${\rm;}
\item[\rm(c)] $e^{t/2}k(e^t)$ is positive definite{\rm;}
\item[\rm(d)] there exists a symmetric probability measure $\nu$ on $\bR$ such that
\be \label{int-expres}
\Omega_D^k(X)=\int_{-\infty}^\infty D^{-\half+it}XD^{-\half-it}\,d\nu(t)
\ee
for all $D\in\bP_d$ and $X\in\bM_d$ with any $d\in\bN$.
\end{itemize}
\end{thm}

It is a well-known consequence of the Stinespring representation theorem that a CP map
$\Phi$ on the matrix algebra $\bM_d$ can be  represented in the form
$\Phi(A) = \sum_j F_j A F_j^* $ with at most $d^2$ matrices $F_j \in \bM_d$ (see, e.g.,
\cite{Kraus,Choi}, \cite[Proposition 4.7]{Paul} or \cite[Appendix~A]{KMNR}).
Thus, for any fixed $D \in \bP_d$, when $\Omega_D^k$ is CP, one can find matrices  $F_j$
in $\bM_d$ such that $\Omega_D^k(X) = \sum_{j = 1}^m  F_j X F_j^*$ with $m \leq d^2$.
But, for fixed $k \in \cK$, the representation will change with $D$ (hence with $d$).
(Even for fixed $D$ the $F_j$ in the representation  are only determined up to a unitary
transformation $F_j \mapsto \sum_i  u_{ij} F_i $ with $u_{ij}$ entries of a unitary
matrix.)  However, when we are allowed to use integral representation, Theorem \ref{thm:CP2}
says that we have the standard representation given in \eqref{int-expres}, from which
the CP of the map $\Omega_D^k$ is directly seen. Moreover, one sometimes has different
integral expressions of $\Omega_D^k$ or $(\Omega_D^k)^{-1}$; a typical example is
\eqref{BKM} for $\Omega_D^k$ in case of $k(x)=\log x/(x-1)$ (see Appendix \ref{app:BKM}).

It follows immediately from \eqref{relK^+K^-} and Theorem~\ref{thm:CP2} that $k\in\cK^-$
if and only if $k\succcurlyeq x^{-1/2}$. Consequently, $x^{-1/2}$ is the largest element
of $\cK^+$ and the smallest of $\cK^-$. Moreover, since $\leqc$ is a partial order on
$\cK$, we conclude

\begin{thm}\label{thm:unique}
The only function in $\cK$ for which both $\Omega_D$ and $\Omega_D^{-1}$ are CP for every
$D\in\bP_d$, $d\in\bN$, is $x^{-1/2}$.
\end{thm}

It follows from Theorem \ref{thm:CP2} that the   problem of determining whether or not
$k\in\cK$ belongs to $\cK^+$ can be reduced to the computation of the Fourier transform of
the function $e^{t/2}k(e^t)$.  However, this is often a hard task as will be seen in
Section~\ref{sect:proofs}.    

In contrast to $\cK$, it does not seem easy to find extreme points of $\cK^+$ other than
$x^{-1/2}$ and $2/(1+x)$ which are the largest and the smallest elements of $\cK^+$,
respectively, in the order $\leqc$ as well as the pointwise order.  However,
some new extreme points will be described in Example \ref{Example 4.3} and
Theorem~\ref{thm:newext}. In addition, a natural boundary point will be found in
Example~\ref{ex:WYD} which is conjectured to be an extreme point.

By comparing part (b) of the next result to \eqref{extptord}, one immediately sees that
$\leqc$ is stronger than the pointwise order.

\bigskip

\begin{prop}\label{prop:simpord}  The following relations hold.
\mbox{}
\begin{itemize}
\item[\em(a)] $k_1^\ext(x) = \dfrac{2}{1+x} \leqc k_\nu^\ext~$ and
$~\wh{k}_\nu^\ext\leqc  \dfrac{1+x}{2x}= k_0^\ext(x)~$ for all $\nu\in[0,1]$.
\item[\rm(b)] $k_\nu^\ext\not\leqc \dfrac{1+x}{2x} = k_0^\ext(x)$ and
$~k_1^\ext(x) = \dfrac{2}{1+x} \not\leqc\wh{k}_\nu^\ext~$ for all $\nu \in (0,1)$.
\item[\rm(c)] $2/(1+x)\leqc x^{-1/2}\leqc(1+x)/2x$.
\end{itemize}
\end{prop}

\begin{proof}
A straightforward computation gives
\begin{linenomath} \begin{align*}
\frac{ k_1^\ext (e^t)}{k_\nu^\ext(e^t) } 
&=\frac{ \wh{k}_\nu^\ext(e^{-t}) }{k_0^\ext(e^{-t}) }
={4\over(1+\nu)^2}\cdot{(e^t+\nu)(1+\nu e^t)\over(1+e^t)^2}  \\
&={4\over(1+\nu)^2}\cdot{\nu(e^t+e^{-t}+2)+(1-\nu)^2\over e^t+e^{-t}+2} \nn  \\
&={4\nu\over(1+\nu)^2}+{2(1-\nu)^2\over(1+\nu)^2}\cdot{1\over\cosh t+1}
\end{align*} \end{linenomath}
from which (a) follows by using $\beta = 1$ in Lemma \ref{lemm:hyperb}\,(2).  Similarly   
\begin{linenomath} \begin{align*}
\frac{ k_\nu^\ext (e^t)  }{ k_0^\ext(e^t) } 
& =\frac{ k_1^\ext (e^{-t})}{ \wh{k}_\nu^\ext(e^{-t}) }
=(1+\nu)^2\,{e^t\over(e^t+\nu)(1+\nu e^t)} \\  \nn 
&=(1+\nu)^2\,{1\over\nu(e^t+e^{-t})+1+\nu^2} \\
&={(1+\nu)^2\over2\nu}\cdot{1\over\cosh t+{1+\nu^2\over2\nu}}.
\end{align*} \end{linenomath}
Since  $(1+\nu^2)/2\nu > 1$  for $\nu\in(0,1)$, this proves (b) by
Lemma \ref{lemm:hyperb}\,(2) again. Finally (c) follows easily from
$$
e^{t/2} k_1^\ext(e^t) = \frac{e^{-t/2}}{k_0^\ext(e^t) }
= \frac{2e^{t/2} }{ e^t+1 } = \frac{1}{ \cosh(t/2)}.
$$
\end{proof}   

\bigskip
Proposition~\ref{prop:lim0} implies that for every $k \in \cK$,  $x k(x) $ is bounded on
$(0, b)$  and $k(x)$ is bounded on $(a,\infty)$ for any $a, b > 0$.  Theorem~\ref{thm:CP2}
implies that a necessary condition for $k\in\cK^+$  is the stronger property that $x^{1/2}k(x)$ is bounded
on $(0,\infty)$. However, this is not a sufficient condition.  Indeed, it holds for all
$k_\nu^\ext(x)$ with $\nu \in (0, 1]$.  Yet, as  will be seen in  Example \ref{Example 4.1}  
$k_\nu^\ext(x) \in \cK^+$ only for $\nu = 1$.  
The following result  will be used in Example \ref{Example 4.4} to
analyze convex combinations
of $x^{-1/2}$ and $k_\nu^\ext$.

\begin{prop}\label{Proposition 3.6}
Assume that $k\in\cK\setminus\cK^+$ and $\lim_{x\to\infty}x^{1/2}k(x)=0$. Then for every
$\lambda\in(0,1]$,
$$
\lambda k(x) +(1-\lambda)x^{-1/2}\notin\cK^+.
$$
\end{prop}

\begin{proof}
Assume that $k\in\cK$ satisfies $\lim_{x\to\infty}x^{1/2}k(x)=0$ and
$\lambda k(x) +(1-\lambda)x^{-1/2}\in\cK^+$ with some $\lambda\in(0,1]$. Then, thanks to
Theorem \ref{thm:CP2} and Bochner's theorem there is a probability measure $\mu$ on $\bR$
satisfying
$$
\lambda e^{t/2}k(e^t)+(1-\lambda)=\hat\mu(t)
\equiv \int_{-\infty}^\infty e^{its}\,d\mu(s),\qquad t\in\bR.
$$
However, the symmetry condition $xk(x)=k\ofxinv$ implies $e^{t/2}k(e^t)=e^{-t/2}k(e^{-t})$,
$t\in\bR$, and hence $\lim_{|t| \to \infty} e^{t/2}k(e^t)=0$.  Therefore, we have
$$
\mu(\{0\})=\lim_{|t| \to \infty}\hat\mu(t)=1-\lambda
$$
(see \cite[Corollary A.8]{HK2}).  This means $e^{t/2}k(e^t)=\hat\mu_0(t)$, $t\in\bR$, 
with the probability measure $\mu_0=\lambda^{-1}(\mu-\mu(\{0\})\delta_0)$, implying the
contradiction  $k\in\cK^+$.
\end{proof}

\section{Examples}   \label{sect:examps}

In this section we list known families of functions in $\cK$ and investigate which functions
in those families belong to $\cK^+$ (or $\cK^-$). In this way we will see that $\cK^+$
indeed contains a variety of functions even though it occupies only a small part of $\cK$.

\subsection{Extreme points and simple averages}  \label{sect:exampav}

\begin{example}\label{Example 4.1}\rm(Extreme points of $\cK$)\quad
The extreme points of $\cK$ are $k_\nu^\ext$, $\nu\in[0,1]$, given in \eqref{extpt}. These
are not in $\cK^+$ unless $\nu=1$ for which we have $k_1^\ext(x)=2/(1+x)$.
Indeed, for $\nu\in(0,1]$ we find
$$
e^{t/2}k_\nu^\ext(e^t)
={(1+\nu)^2\over2\nu}\cdot{\cosh(t/2)\over\cosh t+{1+\nu^2\over2\nu}}.
$$
If $e^{t/2}k_\nu^\ext(e^t)$ is positive definite, then so is its product with the positive
definite $ 1/ \cosh(t/2)$.  But this yields (up to a constant) a function of the form in 
Lemma~\ref{lemm:hyperb}\,(2), which is not positive definite
for $\beta = (1+\nu^2)/2 \nu > 1 $ when $\nu \in (0,1)$.

It was shown in \cite[Example 9]{BP} that $k_\nu^\ext(x)\le x^{-1/2}$ (in the
pointwise order) for all $x>0$ if and only if $3-2\sqrt2\le\nu\le1$.  This example
provides another demonstration that  the $\leqc$ order is stronger and  $\leqc x^{-1/2}$ 
is the  key to determining whether or not a function $k \in \cK^+$.
\end{example}

\begin{example}\label{Example 4.2}\rm
(Convex combinations involving $k_0^\ext$)\quad
Consider the convex combinations
$$
a_{1,0,\lambda}(x) \equiv \lambda k_0^\ext(x)+(1-\lambda) k_1^\ext(x)
= \lambda \frac{1+x}{2x} + (1-\lambda) \frac{2}{1+x}, \qquad\lambda\in[0,1],
$$
of the smallest element of $\cK^+$ and the largest element of $\cK$.  Since
$$
e^{t/2} a_{1,0,\lambda}(e^t) = \frac{1 - \lambda}{\cosh(t/2)} + \lambda \cosh(t/2)
$$
is unbounded for any $\lambda \in (0,1]$, it cannot be positive definite and hence
combining an arbitrarily small amount of $k_0^\ext $ (the largest element of $\cK^{-}$)
with the smallest element of $\cK^+$ moves out of $\cK^+$.   

A similar argument can be used to show that any $k \in \cK$ for which the measure $m$ in
\eqref{intrepnu} has the property that $m(\{ 0 \} ) > 0$ cannot be in $\cK^+$.
\end{example}

\begin{example}\label{Example 4.3}\rm
(Convex combinations of $k_1^\ext$ and $k_\nu^\ext$)\quad
Replacing $k_0^\ext$ in the previous example with another $k_\nu^\ext$ does sometimes
yield convex combination in $\cK^+$. To be precise, let
\be  \label{afunc}
a_{1,\nu,\lambda}(x)
\equiv \lambda k_\nu^\ext(x)+(1-\lambda)\,{2\over 1+x}, \qquad\lambda\in[0,1],
\ee
be a convex combination of the smallest $k_1^\ext(x)=2/(1+x)$ of $\cK^+$ and other
extreme points $k_\nu^\ext$ of $\cK$, $\nu\in(0,1)$. We are interested in the problem to
determine $\nu,\lambda$ for which $a_{1,\nu,\lambda}$ belongs to $\cK^+$. Our result is
that for every $\nu\in[0,1)$, $a_{1,\nu,\lambda}$ is in $\cK^+$ if and only if 
\be  \label{convcombcond}
0 \le \lambda \le \frac{ 2\sqrt\nu}{ (1+\sqrt\nu)^2}
= \frac{2}{ \big(\nu^{1/4} + \nu^{-1/4} \big)^2}.
\ee 
Moreover, $a_{1,\nu,\lambda}$ is an extreme point of $\cK^+$ if and only if equality holds
in \eqref{convcombcond}.
  
Since the proofs require some technical results from Section~\ref{sect:proofs}, they are
postponed to Section~\ref{sect:convextpf}. Note that the right-hand side of 
\eqref{convcombcond} is $< \half$ but $a_{1,1,\lambda}(x)=2/(1+x) \in \cK^+$ for
$\lambda\in[0,1]$. Thus, this example exhibits some discontinuous behavior at $\nu = 1$.  

It is straightforward (see the last paragraph of Section 2) to extend these results to
show that the function
$$
a_{1,\nu,\lambda}(x^{-1})  =  \lambda f_\nu^\ext(x)+(1-\lambda)\, \frac{2x}{1+x},
\qquad\lambda\in[0,1),
$$
is in $\cF^-$ if and only if the inequality holds in \eqref{convcombcond} and that
it is an extreme point of $\cF^-$ if and only if equality holds.  
\end{example} 

\begin{example}\label{Example 4.4}\rm
(Extended Heron means)\quad 
Consider the  convex combinations of   $x^{-1/2} $ and extreme
points of $\cK$, i.e.,
\be   \label{heron}
 \lambda k_\nu^\ext(x)+(1-\lambda)x^{-1/2}
\ee
which are sometimes known as Heron means when $\nu = 1$, in which case
\eqref{heron} is obviously in $ \cK^+ $ for all $ \lambda\in[0,1]$.
However, for $\nu = 0 $ the function \eqref{heron} is in $\cK^+$ only for $\lambda = 0$
since $x^{1/2}k_0^\ext(x)$ is unbounded. Furthermore, it follows from   Proposition \ref{Proposition 3.6} 
that for $\nu\in(0,1)$ and $\lambda \neq 0 $ the function in 
 \eqref{heron}  is never in $\cK^+$ because
$k_\nu^\ext\not\in\cK^+$ and $x^{1/2}k_\nu^\ext(x)\to0$ as $x\to\infty$.

Next, consider  \eqref{heron} with $\nu =  0$ as  the convex combination of the largest
and the smallest elements of  $\cK^-$ for $\lambda\in\bigl(0,\half\bigr)$. Since
$$
e^{-t}\biggl(\lambda\,{1+e^{2t}\over2e^{2t}}+(1-\lambda)e^{-t}\biggr)^{-1}
={1\over\lambda}\cdot{1\over\cosh t+{1-\lambda\over\lambda}}
$$
with $(1-\lambda)/\lambda>1$ is not positive definite by Lemma \ref{lemm:hyperb}\,(2), we
have $\lambda\wh{k}_1^\ext(x)+(1-\lambda)x^{-1/2}\not\in\cK^-$, showing that $\cK^-$ is not
convex.  However,  the dual set $\cF^{-} $ is convex and $f\mapsto k=1/f$ transforms
$\cF^-$ to $\cK^-$. Thus, although $\cK^-$ is not convex, harmonic means of functions in
$\cK^-$ are in $\cK^-$.
\end{example}

\subsection{Families of classic functions in $\cK$}   \label{sect:classic}

Examples given so far suggest us that $\cK^+$ is a rather thin subset of $\cK$. Therefore,
it is a bit surprising that we find a number of one-parameter families in $\cK^+$ in the
examples below.  Each of these families shows some type of symmetry and   monotonicity in
the $ \leqc $ order on maximally suitable intervals.  In all these cases, the symmetry
condition  $x k(x) = k \ofxinv $ can be easily checked and  it is rather straightforward
to use the Pick mapping property to  verify that they are in $\cK$.  Although the most
intriguing family is associated with the WYD skew information, it is also rather complex.

\begin{example}\label{ex:heinz}\rm(Heinz type means)\quad
The family of  functions  
\be  \label{heinz}
k_\alpha^\H(x) \equiv {2\over x^\alpha+x^{1-\alpha}}, \qquad \alpha \in [0,1],
\ee
has the dual family
$$
\wh{k}_\alpha^\H(x) \equiv \frac{1}{ k_\alpha^\H \big( x^{-1} \big) }
= {x^{-\alpha}+x^{-1+\alpha} \over2},\qquad0\le\alpha\le1.
$$
which were used in \cite{TKRWV}.
One easily recovers the Heinz type means via \eqref{HKmean} since
$$
{y \over k_\alpha^\H(x/y)} = {x^\alpha y^{1-\alpha}+x^{1-\alpha}y^\alpha \over2},
\qquad \alpha \in [0,1].
$$
In addition to $k_{1/2}^\H(x)=x^{-1/2} = \wh{k}_{1/2}^\H(x)$, important special cases are
$$
k_0^\H(x) = k_1^\H(x) = \frac{2}{1+x} = k_1^\ext(x), \qquad
\wh{k}_0^\H(x) =  \wh{k}_1^\H(x) = \frac{ 1+x}{2x} =  k_0^\ext(x)
$$
reflecting the obvious symmetry around $x = \half$.
 
Since $e^{t/2}k_\alpha^\H(e^t)=1/\cosh\bigl(\bigl(\alpha-\half\bigr)t\bigr)$ is
positive definite, $k_\alpha^\H\in\cK^+$ for any $\alpha\in[0,1]$ and
$\wh{k}_\alpha^\H\in\cK^-$ for any $\alpha\in[0,1]$. A different proof of the former was
in \cite[Example 3]{BP}.

If $ 0 \le\alpha\le\beta\le \half$, then $k_\alpha^\H\leqc k_\beta^\H$
(see \cite[Section 2]{HK1}) so that the pair of functions $k_\alpha^\H$ for
$\alpha \in \big[0, \half\big]$ and $\wh{k}_\alpha^\H$ for $\alpha \in \big[\half, 1\big]$
can be regarded as a single family which increases in the $\leqc$ order from the smallest
to the largest element of $\cK$. 

Moreover, whenever $0\le\alpha\le\beta\le\half$,
$$
{k_\alpha^\H(e^t)\over k_\beta^\H(e^t)}
={\cosh\left(\left({1\over2}-\beta\right)t\right)
\over\cosh\left(\left({1\over2}-\alpha\right)t\right)}
$$
is infinitely divisible by
\cite[Theorem 1]{BK}. 
\end{example}

\begin{example} \label{ex:binom}  \rm(Binomial means or power means)\quad
The functions
$$
k_\alpha^\B(x)\equiv  \bigg( \frac{2}{x^\alpha+1} \bigg) ^{1/\alpha},
\qquad\alpha\in[-1,1],
$$
are easily verified to be in $\cK$ as observed in \cite[Theorem 3\,(i)]{Na} and correspond
to the binomial (or power) means 
$$
{y \over k_\alpha^\B(x/y)} = \bigg({x^\alpha+y^\alpha\over2}\biggr)^{1/\alpha},
\qquad\alpha\in[-1,1].
$$
Interesting special cases are
\begin{align*}
&k_{-1}^\B(x)=\frac{1+x}{2x} =k_0^\ext(x), \qquad
\ \ k_{-1/2}^\B(x) = \frac{(1 + \sqrt{x})^2}{4x} = \wh{k}_{1/2}^\WYD(x), \\
&k_0^\B(x)= \ds{ \lim_{\alpha \raw 0 } } \, k_\alpha^\B(x) = x^{-1/2}, \qquad
k_{1/2}^\B(x) = {4\over(1 + \sqrt{x})^2} = k_{1/2}^\WYD(x), \\
&k_1^\B(x)={2\over 1+x}=k_1^\ext(x),
\end{align*}
where $k_p^\WYD$ is given in Example \ref{ex:WYD}.
Moreover,  $k_{\alpha}^\B(x)= \wh{k}_{-\alpha}^\B(x)$  which implies that for this family
$$
(\Omega_D^\alpha)^{-1} = \Omega_{D^{-1}}^{-\alpha}, \qquad \alpha\in[-1,1],
$$
with the obvious abuse of notation. It follows from \cite[Theorem 9]{Ko3} that if
$-1\le \beta \le \alpha \le 1$ then $k_\alpha^\B\leqc k_\beta^\B$, so that we have a
decreasing family in the $ \leqc $ order. Since $k_0^\B(x)=x^{-1/2}$, we conclude
\begin{itemize}
\item $k_\alpha^\B\in\cK^+$ if and only if $\alpha \in [0,1]$, 
\item $k_\alpha^\B\in\cK^-$ if and only if $\alpha \in [-1,0]$.
\end{itemize}  
Moreover, $k_\alpha^\B(e^t)/k_\beta^\B(e^t)$ is infinitely divisible whenever
$\beta\le\alpha$ \cite[Theorem 9]{Ko3}.
\end{example}

\begin{example}\label{ex:ALG}\rm(Power difference means)\quad
The family of functions
$$
k_\alpha^\PD(x)\equiv \frac{\alpha}{\alpha-1} \cdot \frac{x^{\alpha-1}-1}{x^\alpha-1},
\qquad\alpha\in[-1,2],
$$
gives the family of power difference means considered in \cite{HK1,HK2}. In fact,
\begin{equation}\label{(4.5)} 
{y \over k_\alpha^\PD(x/y)} = M_\alpha(x,y) \equiv
\frac{\alpha-1}{\alpha} \cdot {x^\alpha-y^\alpha\over x^{\alpha-1}-y^{\alpha-1}},
\end{equation}
whose family is also called the A-L-G interpolation means since it interpolates the
arithmetic, the logarithmic and the geometric means by allowing us to recover all of these
as special cases
\begin{align*}
&k_{-1}^\PD(x) = {1+x\over2x}, \qquad
k_0^\PD(x)=\lim_{\alpha\to0}k_\alpha^\PD(x)={x-1\over x\log x}, \\
&k_{1/2}^\PD(x) = x^{-1/2}, \qquad
\ k_1^\PD(x) = \lim_{\alpha\to 1}k_\alpha^\PD(x)={\log x\over x-1}, \\
&k_2^\PD(x)= {2\over 1+x}.
\end{align*}
It is known \cite[Proposition 4.2]{HK1} that $k_\alpha^\PD\in\cK$ for all
$\alpha\in[-1,2]$. Moreover, we have $k_\alpha^\PD = \wh{k}^\PD_{1-\alpha}$, which implies
that for this family
$$
(\Omega_D^\alpha)^{-1}  = \Omega_{D^{-1}}^{1-\alpha}, \qquad  \alpha\in[-1,2],
$$
with  the obvious abuse of notation.  If $-1\le \beta \le \alpha \le2$ then
$k_\alpha^\PD\leqc k_\beta^\PD$ (see \cite[Theorem 2.1]{HK1}), so that we have another
increasing family.  Thus, since $k_{1/2}^\PD(x)=x^{-1/2}$, we can conclude
\begin{itemize} 
\item $k_\alpha^\PD\in\cK^+$ if and only if $\alpha\in\bigl[\half,2\bigr]$,
\item $k_\alpha^\PD\in\cK^-$ if and only if $\alpha\in\bigl[-1,\half\bigr]$.
\end{itemize} 
Moreover, the monotonicity can be strengthened to the infinite divisibility of
$k_\alpha^\PD(e^t)/\allowbreak k_\beta^\PD(e^t)$ for $\beta\le\alpha$ by \cite[Theorem 5]{Ko3}.
\end{example}

\begin{example} \label{ex:WYD} \rm(WYD family)\quad
One of the best known families in $\cK$ is an outgrowth of the Wigner-Yanase-Dyson skew
information discussed in Section~\ref{sect:back} which leads to the functions
\begin{equation}\label{WYD}
k_p^\WYD(x)\equiv {1\over p(1-p)}\cdot{(1-x^p)(1-x^{1-p})\over(1-x)^2},\qquad p\in[-1,2].
\end{equation}
This family is symmetric around $p=\half$, and the special cases $p=0,1$ should be
understood by continuity, i.e.,
$$
k_1^\WYD(x) =k_0^\WYD(x) =  \lim_{p\to1}k_p^\WYD(x)={\log x\over x-1}.
$$
Other important special cases are
\begin{align*}
&k_{1/2}^\WYD(x)={4\over(1+\sqrt x)^2},\qquad
k_{-1/2}^\WYD(x)=k_{3/2}^\WYD(x)={4\over3}\cdot{1+\sqrt x+x\over\sqrt x(1+\sqrt x)^2}, \\
&k_{-1}^\WYD(x)=k_2^\WYD(x)={1+x\over2x}=k_0^\ext(x).
\end{align*}

We can summarize the CP situation for this family as follows:
\begin{itemize}
\item[(a)] $ k_p^\WYD  \in \cK^+$ if and only if $p \in [0,1]$,
\item[(b)]  $ k_p^\WYD  \in \cK^-$ if and only if
$p \in \big[-1,-\half \big] \cup \big[\tfrac{3}{2}, 2 \big]$.
\end{itemize}

For $p \in \big[\half, 2 \big]$ the functions $k_p^\WYD $ increase monotonically   
with respect to the $ \leqc $ order.  Set $r\equiv p+q-1$, $\alpha\equiv p/r$,
$\beta\equiv q/r$ so that $r>0$ and $0<\alpha<\beta$. We note
\begin{linenomath} \begin{align*}
{k_p^\WYD(x)\over k_q^\WYD(x)}
&={q(1-q)\over p(1-p)}\cdot{(1-x^p)(1-x^{1-p})\over(1-x^q)(1-x^{1-q})} \\
&={\beta(\alpha-1)\over\alpha(\beta-1)}\cdot
{(1-x^{r\alpha})(1-x^{r(\beta-1)})\over(1-x^{r\beta})(1-x^{r(\alpha-1)})}
={M_\alpha(x^r,1)\over M_\beta(x^r,1)},
\end{align*} \end{linenomath}
where $M_\alpha(x,y)$ is the power difference mean defined by \eqref{(4.5)} (for any real
parameter $\alpha$). Therefore, when $\half\le p\le q\le2$, $k_p^\WYD(e^t)/k_q^\WYD(e^t)$
is infinitely divisible by \cite[Theorem 5]{Ko3} and in particular
$k_p^\WYD\leqc k_q^\WYD$.

Thus, the functions $k_p^\WYD$ form a smooth family which are in $\cK^+$ up to $p=1$ when
$p$ increases from $\half$. Therefore, $k_1^\WYD$ lies on the boundary of $\cK^+$, and we
conjecture that it is an extreme point of $\cK^+$.
 
The operator $\Omega_D^k$ for $k=k_1^\WYD$ is given by \eqref{BKM}, which implies that
$k_p^\WYD \in \cK^+$ for $p = 0,1$.  In the proof of \cite[Theorem~2]{BP},
explicit (double) integral representations were obtained for $\Omega_D^k$ when
$k= k_p^\WYD$ for $p \in (0,1)$ in such a way that the CP of $\Omega_D^k$ immediately
follows and hence  $ k_p^\WYD \in \cK^+$ for $p \in (0,1)$. This gives the ``if'' part of
(a). An alternate proof of this, as well as details for the remaining claim above are
given in Section~\ref{sect:WYD}.  This requires results from Section~\ref{sect:sinh} which
are of independent interest.

Unlike other families we consider, the functions
$\wh{k}^\WYD_p(x) = 1/k_p^\WYD\big( {x}^{-1} \big)$  do not belong to the WYD family.
Despite the extensive study of WYD metrics, there seems to have been little attention
given to this dual family
$$
\wh{k}_p^\WYD(x) \equiv p (1-p) \frac{(1-x)^2}{(x - x^{1-p})(x - x^p)},
\qquad p\in[-1,2].
$$
This is symmetric around $p=\half$ and special cases are
$$
\wh{k}_{1/2}^\WYD(x)={(1+\sqrt x)^2\over4x},\qquad
\wh{k}_1^\WYD(x)={x-1\over x\log x},\qquad
\wh{k}_2^\WYD(x)={2\over 1+x}=k_1^\ext(x).
$$
By \eqref{relK^+K^-} and (b) above the functions $\wh{k}_p^\WYD$ are in $\cK^+$ for
$p \in \big[ \frac{3}{2}, 2\big]$ and in $\cK^-$  for $p \in [0,1]$.
\end{example}

\begin{example} \label{ex:Stl} \rm(Stolarsky means)\quad
As in the WYD example above,  the dual of the Stolarsky family
gives a different family.   In this case, we introduce both  
$$
k_\alpha^\St(x) \equiv
\bigg( \frac{ x^\alpha -1 }{\alpha (x - 1)} \bigg)^{1\over 1 - \alpha}, \qquad 
\wh{k}_\alpha^\St(x) \equiv
\biggl({x^{1-\alpha}-x\over\alpha(1-x)}\biggr)^{1\over\alpha-1}, \qquad \alpha \in [-2,2].
$$
It is known \cite[Theorem 3\,(iii)]{Na} (also \cite[Theorem 3]{BP}) that
$k_\alpha^\St \in \cK $ for $\alpha\in[-2,2]$ and this range for $\alpha$ such that
$k_\alpha^\St\in\cK$ is optimal. The functions  $k_\alpha^\St(x) $ correspond to the
familiar family of Stolarsky means as
\be   \label{Stmean}
{y \over k_\alpha^\St(x/y)} = y \, \wh{k}_\alpha^\St(y/x) = S_\alpha(x,y) \equiv
 \bigg( \frac{ x^\alpha-y^\alpha  }{ \alpha(x-y)}\bigg)^{1\over\alpha -1 }.
\ee 
The mean $S_1(x,y)=e^{-1}(x^x/y^y)^{1/(x-y)}$ for $\alpha = 1$ is called the identric mean. 

The functions $k_\alpha^\St$ include more familiar special cases than $\wh{k}_\alpha^\St$
as follows:
\begin{align*}
&k_2^\St(x) = \frac{2}{1+x}, \qquad
k_1^\St(x) =\lim_{\alpha\to1}k_\alpha^\St(x )=  e \,x^{x\over1-x}, \qquad
k_{1/2}^\St(x) = \frac{4}{(1+\sqrt{x})^2}, \\
&k_0^\St(x)=\lim_{\alpha\to0}k_\alpha^\St(x)=\frac{\log x}{x-1}, \qquad
k_{-1}^\St(x)=x^{-1/2}, \qquad k_{-2}^\St(x) = \bigg( \frac{1+x}{2x^2} \bigg)^{1/3}, 
\end{align*}
which provide an interesting comparison with the other families, as shown in Table 1.

 When
$$
S_{\alpha,\beta}(x,y) \equiv
\biggl({\beta(x^\alpha-y^\alpha)\over\alpha(x^\beta-y^\beta)}\biggr)^{1\over\alpha-\beta},
$$
it was proved in \cite[Theorem 12]{Ko3} that
$S_{\alpha,\beta}(e^x,1)/S_{\alpha',\beta'}(e^x,1)$ is infinitely divisible as long as
$\alpha\le\alpha'$ and $\beta\le\beta'$. Since in particular
$S_{\alpha,1}(x,y)=S_\alpha(x,y)$, it follows from \eqref{Stmean} that this implies that
when $\alpha\le\beta$ the dual family  $\wh{k}_\alpha^\St\leqc \wh{k}_\beta^\St$ increases
and $k_\beta^\St\leqc k_\alpha^\St$  decreases.   We can then conclude that
\begin{itemize}
\item $k_\alpha^\St\in\cK^+$ and $\wh{k}_\alpha^\St\in\cK^-$  if and only if
$\alpha\in[-1,2]$,
\item $k_\alpha^\St\in\cK^-$ and $\wh{k}_\alpha^\St\in\cK^+$ if and only if
$\alpha\in[-2,-1]$.
\end{itemize}

As remarked above, the dual functions form a different family with special cases
\begin{align*}
&\wh{k}_{-2}^\St(x)= \bigg( \frac{2}{x(1+x)} \bigg)^{1/3}, \qquad\quad
\wh{k}_{-1}^\St(x)=x^{-1/2}, \qquad\qquad
\wh{k}_0^\St(x)=\lim_{\alpha\to0}\wh{k}_\alpha^\St(x)=\frac{x-1}{x \log x}, \\
&\wh{k}_{1/2}^\St(x)={(1+\sqrt x)^2\over4x}, \qquad 
\wh{k}_1^\St(x)=\lim_{\alpha\to1}\wh{k}_\alpha^\St(x)= e^{-1} \,x^{1\over x-1}, \qquad 
\wh{k}_2^\St(x)={1+x\over2x}=\wh{k}_0^\ext(x).
\end{align*}
The pair $k_{1-\alpha}^\St$ for  $-1\le\alpha\le2$ and $\wh{k}_{\alpha}^\St$ for
$-1\le\alpha\le2$ can be regarded as a single family which increases in the $\leqc$ order
from $k_0^\ext$ to $k_1^\ext$.  The functions  $k_1^\St$ and $ \wh{k}_{-2}^\St$  give
new members of  $\cK^+$  which do not not appear in any of the other families.  
Moreover,   $ \wh{k}_{-2}^\St$  must lie on the boundary of both $\cK^+$ and  $\cK$,  
which implies that   $\cK^+$
touches the boundary of $\cK$ at the interior of a face.  It seems reasonable to conjecture
that  $ \wh{k}_{-2}^\St$ is an extreme point of $\cK^+$.
\end{example}

It is interesting to compare the behavior of these  examples as the
parameters $\alpha$ and $p$ change, as summarized in Table~\ref{tab1} and
Figure~\ref{fig1}.

\begin{table}[h] \centerline{
\begin{tabular} {|c | c | c | c |c |} \hline
& 4.6\ \ $k_\alpha^\B$ & 4.7\ \ $k_\alpha^\PD$ &
4.8\ \ $k_p^\WYD$ & 4.9\ \ $k_\alpha^\St$ \\  \hline
$\frac{2}{1+x} $& 1 & 2 & & 2  \\   \hline
$e x^{x/(1-x)} $ & & &  & 1 \\    \hline
$\frac{4}{(1 + \sqrt{x} )^2} $ &  $\half$ & & $\half$ & $ \half$ \\   \hline
$\frac{\log x}{x-1} $ & & 1 & 0 & 0 \\   \hline
$x^{-1/2} $ & 0 & $\half$ && $-1$ \\  \hline
\end{tabular} }
\caption{Summary of common crossing points}   \label{tab1} 
\end{table}    
The Stolarsky family is the only one which goes through all
of the indicated points.   The WYD family is the only one which 
does not begin and end at the smallest and largest elements,
and moves outside of both $\cK^+$ and $\cK^-$ for some
parameter range.

\begin{figure}[h]
\centerline{\includegraphics[height=8cm]{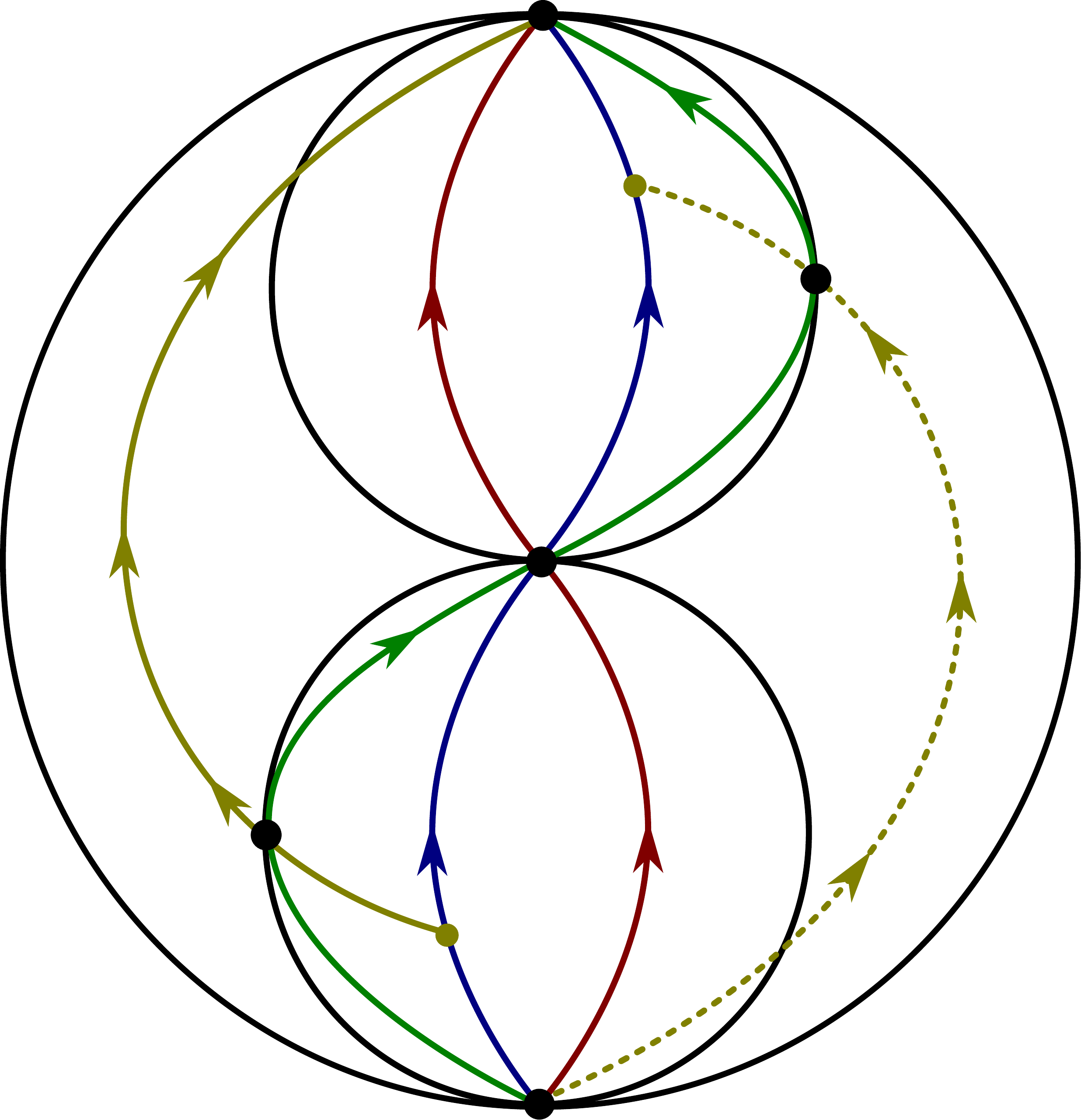}}
\caption{Schematic diagram of families in $\cK$ parameterized so that 
they  increase in the $\leqc$ order.
The lower ball corresponds to $\cK^+$ and the upper ball  to
$\cK^-$.  The three curves inside $\mathcal{K}^+\cup\mathcal{K}^-$
beginning at the smallest member $2/(1+x)$ are described from right to left.
The rightmost curve (red) describes the Heinz family $k_\alpha^{\mathrm{H}}$
($0\le\alpha\le1/2$) and $\widehat k_\alpha^{\mathrm{H}}$ ($1/2\le\alpha\le1$);
the next (blue) curve the binomial family $k_{-\alpha}^{\mathrm{B}}$
($-1\le\alpha\le1$); the next (green) curve the power difference family
$k_{-\alpha}^{\mathrm{PD}}$ ($-2\le\alpha\le1$).
The brown curve on the left the WYD family $k_p^\WYD$ in the
range  $p \in [\half,2]$ and the dotted brown curve on the right the dual WYD family.
 The crossings at  $4/(1 + \sqrt{x} )^2 $ and  $\log x/(x-1) $ can easily be seen.
The  complex Stolarsky family, which is the only one which 
starts at  the smallest $2/(1+x)$ and goes through both of these
crossings while remaining in $\cK^+$ before reaching $x^{-1/2} $, is not shown.}   \label{fig1}
  \end{figure} 

\subsection{Geometric bridges}  \label{sect:geombrdg}

In Examples \ref{Example 4.2} and \ref{Example 4.3} of Section \ref{sect:exampav} we
considered arithmetic weighted averages of $2/(1+x)$ (the smallest of $\cK^+$) or
$x^{-1/2}$ (the largest of $\cK^+$) with extreme points of $\cK$, and noticed that such
averages can be in $\cK^+$ in rather limited cases. In this section we consider
a different type of averages, often called a {\it geometric bridge}, which is defined as
weighted geometric means $[k_1(x)]^{1-\lambda}[k_2(x)]^\lambda$, $0\le\lambda\le1$, of
$k_1,k_2\in\cK$. We first show that $\cK$ and $\cK^\pm$ are all closed under geometric
bridge interpolations as far as some infinite divisibility condition is satisfied for
$\cK^\pm$. The equivalence of (ii) and (iii) in the next theorem implies that a similar
result holds for $\cF$ and $\cF^\pm$.
   \pagebreak
\begin{prop}\label{Proposition 5.1}
If $k_1,k_2\in\cK$, then for every $\lambda\in[0,1]$ the function
$[k_1(x)]^{1-\lambda} [k_2(x)] ^\lambda$ is also in $\cK$. Moreover, if $k_1,k_2\in\cK^+$
{\rm(}resp., $\cK^-${\rm)} and one of the following conditions is satisfied, then for every
$\lambda\in[0,1]$ the function $[k_1(x)]^{1-\lambda}[k_2(x)]^\lambda$ is also in $\cK^+$
{\rm(}resp., $\cK^-${\rm)}\,$:$
\begin{itemize}
\item[\rm(i)] both $e^{t/2}k_1(e^t)$ and $e^{t/2}k_2(e^t)$ are infinitely divisible,
\item[\rm(ii)] $k_2(e^t)/k_1(e^t)$ is infinitely divisible,
\item[\rm(iii)] $k_1(e^t)/k_2(e^t)$ is infinitely divisible.
\end{itemize}
\end{prop}

\begin{proof}    
To prove the first assertion, let $k_1,k_2\in\cK$; then by Theorem \ref{thm:kequiv} 
they have the Pick mapping property, from which it follows that
$[k_1(x)]^{1-\lambda}[k_2(x)]^\lambda$ also has this property and hence is operator
monotone decreasing. Since the symmetry condition in Theorem \ref{thm:kequiv} is obvious,
we conclude that $k_1^{1-\lambda}k_2^\lambda \in \cK$.

To prove the second assertion, let $k_1,k_2\in\cK^+$ and $0<\lambda<1$. When (i) is
satisfied, $[e^{t/2}k_1(e^t)]^{1-\lambda}$ and $[e^{t/2}k_2(e^t)]^\lambda$ are
positive definite and hence so is the product
$e^{t/2}[k_1(e^t)]^{1-\lambda}[k_2(e^t)]^\lambda$. We note 
$$
e^{t/2}[k_1(e^t)]^{1-\lambda}[k_2(e^t)]^\lambda
=e^{t/2}k_1(e^t)\biggl({k_2(e^t)\over k_1(e^t)}\biggr)^\lambda
=e^{t/2}k_2(e^t)\biggl({k_1(e^t)\over k_2(e^t)}\biggr)^{1-\lambda}
$$
so that we get the desired positive definiteness from either (ii) or (iii).

Finally, the assertion for $\cK^-$ is easily verified by taking
$\wh{k}_j(x)\equiv 1/k_j\ofxinv$ and using \eqref{relK^+K^-}.
\end{proof}

Recall that all the one-parameter families in $\cK^+$ given in Examples
\ref{ex:heinz}--\ref{ex:Stl}
satisfy the property of infinite divisibility (an order stronger
than $\leqc$). Therefore, the above proposition implies that geometric bridges
joining $k_1,k_2$ in each of these family sits inside $\cK^+$. 

\begin{example}\rm
Consider the bridge
$$
k_\alpha(x) \equiv  [ k_\alpha^\St(x) ]^{1 - \alpha} [x^{-1/2}]^\alpha
= \frac{x^{\alpha/2} - x^{-\alpha/2}}{ \alpha ( x-1)}.  
$$
By Proposition~\ref{Proposition 5.1} this is in $\cK^+$ for $\alpha \in [0,1]$.  In fact, 
$k_\alpha \in \cK$ in the larger range $\alpha \in [0,2]$.  One way to see this is to
observe that  $g_\beta(x) = x^{-\beta}(1-x) = x^{-\beta} - x^{1-\beta}$
is operator convex for $\beta \in [0,1]$.  Then  it follows from \cite[Theorem II.13]{LR} that
$$
\frac{g_\beta(x) +  x g_\beta(x^{-1} ) }{(x-1)^2}
= \frac{ x^\beta - x^{-\beta}}{x - 1} =  2 \beta  \, k_{2 \beta}(x)
$$ 
is a multiple of a function in $\cK$ for $\beta \in [0,1]$. 

Since
$$
e^tk_\alpha(e^{2t})={1\over\alpha}\cdot{\sinh(\alpha t)\over\sinh t},
$$
it is easy to see by Lemma \ref{lemm:hyperb}\,(1) that $k_\alpha\in\cK^+$ if and only if
$\alpha\in[0,1]$ while $k_\alpha\in\cK^-$ if and only if $\alpha\in[1,2]$.   
It is also
known \cite[Theorem 2]{BK} that $k_\alpha(e^t)/k_\beta(e^t)$ is infinitely divisible
whenever $\alpha\le\beta$.  Given the special cases
\begin{align*}
&k_0(x)=\lim_{\alpha\to0}k_\alpha(x)={\log x\over x-1},\qquad
k_{1/2}(x)={2\over x^{1/4}+x^{3/4}}=k_{1/4}^\H(x), \\
&k_1(x)=x^{-1/2},\qquad\qquad\qquad\quad\ \ k_2(x)={1+x\over2x}=k_0^\ext(x),
\end{align*}
it follows that  $k_\alpha(x)$ is a family which increases on $[0,2] $ in the $\leqc$ order
from  $\log x/(x-1)$ to $(1+x)/2x $.

The connection between $g_\beta$ and $ k_{2 \beta} $ is interesting because, as mentioned in Section~\ref{sect:back},
$g(x) = (x-1)^2 k(x) $ is always an operator convex function  with the properties needed
to define a symmetric quasi-entropy.   Although one can begin with a function $g(x)$ which
does {\em not}  satisfy $g(x) = x g(x^{-1}) $ and generate a function   $k \in \cK$, it is not at all
obvious how to reverse the process without obtaining a  symmetric $g$.   In this case, we have  
found an asymmetric $g$, in particular $g_{1/2}(x) = x^{-1/2}  - x^{1/2} $,  which
generates the key function $k(x) = x^{-1/2} \in \cK$.   The associated quasi-entropies do not seem
to have been studied previously, but appeared recently in \cite{Rusk2}.

\end{example}

The remaining examples are
concerned with geometric bridges joining $k_1^\ext$ and other extreme points of $\cK$
which require a more difficult analysis.

\begin{example}\label{ex:geombrdg}\rm
For $\mu,\nu,\lambda\in[0,1]$ we define
$$
g_{\mu,\nu,\lambda}(x)\equiv k_\mu^\ext(x)^{1-\lambda}k_\nu^\ext(x)^\lambda
=k_\mu^\ext(x)\biggl({k_\nu^\ext(x)\over k_\mu^\ext(x)}\biggr)^\lambda
$$
with $k_{\nu}^{\ext}$ given by (\ref{extpt}). This is in $\cK$ by
Proposition \ref{Proposition 5.1}. A special case
$$
g_{1,0,\lambda}(x)=k_1^\ext(x)\biggl({k_0^\ext(x)\over k_1^\ext(x)}\biggr)^\lambda
=x^{-\lambda}\biggl({2\over1+x}\biggr)^{1-2\lambda},
\qquad0\le\lambda\le1,
$$
was treated in \cite[Example 5]{BP}. We have
$$
e^{t/2}g_{1,0,\lambda}(e^t)=\biggl({1\over\cosh(t/2)}\biggr)^{1-2\lambda},
$$
which is positive definite exactly when $0\le\lambda\le\half$ since $1/\cosh t$ is
infinitely divisible (see \cite[Theorem 1]{BK} for instance). Therefore, $g_{1,0,\lambda}$
is in $\cK^+$ if and only if $0\le\lambda\le\half$.
\end{example}

\begin{example}  \label{ex:geom.gen} \rm
For the more general case 
\begin{equation}\label{(4.7)} 
g_{1,\nu,\lambda}(x)=k_1^\ext(x)\biggl({k_\nu^\ext(x)\over k_1^\ext(x)}\biggr)^\lambda
=\biggl({2\over1+x}\biggr)^{1-2\lambda}
\biggl({(1+\nu)^2\over(x+\nu)(1+\nu x)}\biggr)^\lambda,
\end{equation}
which increase pointwise with $\lambda\in[0,1]$ from $k_1^\ext$ to $k_\nu^\ext$.
Its behavior (in the present context) when $\nu \in (0,1)$ seems much more mysterious.
Our results here are:
\begin{itemize}
\item[(i)] the pointwise order of $g_{1,\nu,\lambda}$ in $\lambda$ can be also strengthened
to the $\leqc$ order, and consequently the set
$ \{\lambda\in[0,1]:g_{1,\nu,\lambda}\in\cK^+\} $
is a subinterval $[0,\lambda_c(\nu)]$,
\item[(ii)] for each $\nu \in (0,1)$ the critical value $\lambda_c(\nu)$ satisfies
\be \label{critical}
{1\over4}\le\lambda_c(\nu)\le{1\over3}.
\ee
\end{itemize} 
The proof requires some lengthy computations of Fourier transforms, which will be
presented in Section 6.3.      Unfortunately, we do not have any   information
about the form of $\lambda_c(\nu) $. 
\end{example}

\begin{example} \label{ex:hans} \rm
It is worth noting that a family of modified bridges
\be \label{Hansen-bridge}
g_{1,1-\lambda,\lambda}(x)=\biggl({2\over1+x}\biggr)^{1-2\lambda}
\biggl({(2-\lambda)^2\over(1+x-\lambda)(1+(1-\lambda)x)}\biggr)^\lambda,
\qquad0\le\lambda\le1,
\ee
joining $k_1^\ext$ and $k_0^\ext$ was constructed in \cite{Ha} for the explicit purpose of
finding a one-parameter family which increases from  $k_1^\ext$ and $k_0^\ext$ in the
pointwise order and all of whose elements except for $\lambda=1$ are regular (here
$k\in\cK$ is regular if $\lim_{x\searrow0}k(x)<+\infty$).  Without the regularity requirement, the families in Examples 4.5--4.7  and $4.9$ have this property in the
stronger $\leqc$ order.

From the same computation as in Lemma \ref{Lemma 6.4} (Section \ref{sect:geombrdgpf}) we
observe that the set $\Lambda\equiv\{\lambda\in[0,1]:g_{1,1-\lambda,\lambda}\in\cK^+\}$
includes $\bigl[0,{1\over4}\bigr]$. Koenraad Audenaert did
some numerical work suggesting that $g_{1,1-\lambda,\lambda}$ is not in $\cK^+$ for
$\lambda \geq 0.3$ giving a CP crossing at a point slightly smaller than $0.3$
which would be consistent with \eqref{critical}.
However, we do not know   strong monotonicity in
the $\leqc$ order for the family \eqref{Hansen-bridge}.  To conclude that
$\Lambda$ is of the form $[0,a]$ we would need a stronger result, e.g.,
that $\lambda_c(\nu) $ is monotone  in $\nu$. 
 \end{example}

\section{Positive definite functions}   \label{sect:proofs}

In this section, we present results on positive definiteness and infinite divisibility of
certain functions involving hyperbolic functions, which are needed in our proofs. The
study here is considered as a continuation of \cite{Ko1,Ko2,Ko3}, which are of independent
interest.

\subsection{Positive definiteness of $\sinh$ ratios}  \label{sect:sinh}

We investigate positive definiteness of the function
$$
f(t)\equiv \frac{\sinh (at)\sinh (bt)}{\sinh^2 t} 
$$
with $a, b >0$. If $a,b \leq 1$, then $f(t)$ is a positive definite function as the
product of two such functions (see Lemma \ref{lemm:hyperb}\,(1)).  It is actually
infinitely divisible as is explained in \cite{BK,Ko1} for instance. We will show that the
converse also holds true.

\begin{thm}\label{Theorem 5.1}
The function $f(t)$ is positive definite if and only if $a,b \leq 1$.
\end{thm}

When $a+b >2$, we have $\lim_{t \to \pm \infty} f(t)=\infty$ so that $f(t)$ cannot be
positive definite. When $a+b=2$ and $a \neq b$, the obvious estimate
$$
f(0)=ab < \left(\frac{a+b}{2}\right)^2=1=\lim_{t \to \pm\infty}f(t)
$$
also shows failure of positive definiteness.

We will assume $a+b < 2$, and we must show that $f(t)$ is not positive definite as long as
$a>1$ (and hence $0 < b <1$).  For this purpose it suffices to deal with $a, b$ rational. 
Indeed, if $f(t)$ were positive definite for such $a,b$ (and the result is known for such
rational parameters), then with $a',b'$ rational satisfying $1 < a' \leq a$ and
$0 < b' \leq b$  the product
$$
f(t)\,\frac{\sinh(a't)\sinh(b't)}{\sinh(at)\sinh(bt)}
=\frac{\sinh(a't)\sinh(b't)}{\sinh^2t}
$$
would be positive definite, a contradiction.

Hence, we will assume that $a,b$ are rational in the rest. Obviously we can further assume
\begin{equation}\label{(5.1)} 
\mbox{$a=\frac{m}{n} > 1$, \ $b=\frac{k}{n} >0$, \ $a+b < 2$ with $n,m,k \in \bN$ even}.
\end{equation}
The most delicate part in our proof for Theorem \ref{Theorem 5.1} is covered in the next
lemma, and the rest of the subsection will be devoted to its proof.
\begin{lemma}\label{Lemma 5.2}
The function $f(t)$ cannot be positive definite for $a,b$ rational described by
$(\ref{(5.1)})$.
\end{lemma}

For a fixed $s \in \bR$ we set
$$
F(z) \equiv f(z)\,e^{isz}
\ \left( 
=
\frac{\sinh (az)\sinh (bz)}{\sinh^2 z} \, e^{isz}
=\frac{\sinh \left(\mbox{{\small $\frac{m}{n}$}}\,z\right)
\sinh \left(\mbox{{\small $\frac{k}{n}$}}\,z\right)}{\sinh^2 z} \, e^{isz}
\right)
$$
with $a,b$ given by (\ref{(5.1)}), and compute its integral along the following
rectangle $\Gamma$:
$$
\begin{array}{lll}
\Gamma_1&z=t,& t: -R \to R,\\
\Gamma_2&z=R+is,& s: 0 \to n\pi,\\
\Gamma_3&z=t+in\pi,& t: R \to -R,\\
\Gamma_4&z=-R+is,& s: n\pi \to 0.
\end{array}
$$

We observe
\begin{linenomath} \begin{align*}
&
\sinh(t+in\pi)=\sinh t,\\
&
\sinh\left(\mbox{{\small $\frac{m}{n}$}}\,(t+in\pi)\right)
=\sinh\left(\mbox{{\small $\frac{m}{n}$}}\,t+im\pi\right)
=\sinh\left(\mbox{{\small $\frac{m}{n}$}}\,t\right),\\
&
\sinh\left(\mbox{{\small $\frac{k}{n}$}}\,(t+in\pi)\right)
=\sinh\left(\mbox{{\small $\frac{k}{n}$}}\,t+ik\pi\right)
=\sinh\left(\mbox{{\small $\frac{k}{n}$}}\,t\right)
 \end{align*} \end{linenomath}
(since $n,m,k$ are even) so that we have $F(t+in\pi)=f(t)\,e^{is(t+in\pi)}$ and
\begin{linenomath} \begin{align}\label{(5.2)} 
\int_{\Gamma_1 \cup \Gamma_3}F(z)\,dz
&=
\int_{-R}^{\,R}f(t)\,e^{ist}dt+\int_{R}^{\,-R}f(t)\,e^{ist}e^{-n\pi s}dt
\notag\\
&=
\left(1-e^{-n\pi s}\right)\int_{-R}^{R}f(t)\,e^{ist}dt 
\notag\\
&=
2e^{-n\pi s/2}\sinh(n\pi s/2)\int_{-R}^{R}f(t)\,e^{ist}dt. 
 \end{align} \end{linenomath}
Since $a+b < 2$, we have $f(z) \to 0$ uniformly on the strip
$\{z \in \bC;\, 0 \leq \Im z \leq n\pi\}$  as $\Re z \to \pm\infty$ and hence
\begin{equation}\label{(5.3)} 
\lim_{R \to \infty} \int_{\Gamma_2 \cup \Gamma_4} F(z)\,dz=0.
\end{equation}
Therefore, the Fourier transform of $f(t)$ can be computed  from $\int_{\Gamma}F(z)\,dz$.
If the Fourier transform fails to be positive, then Lemma \ref{Lemma 5.2} follows from
Bochner's theorem.

Note that $z=0,in\pi$ are zeros of $\sinh^2 z$ of order $2$. However, these two points are
also zeros for 
$\sinh\left(\mbox{{\small $\frac{m}{n}$}}\,z\right)$, 
$\sinh\left(\mbox{{\small $\frac{k}{n}$}}\,z\right)$
so that $z=0,in\pi$ are removable singularities for $F(z)$. The poles (inside of $\Gamma$)
closest to $\Gamma$ are
$$
z_1=i\pi \quad \mbox{and} \quad z_{n-1}=i(n-1)\pi.
$$
Note that $z_1, z_{n-1}$ are zeros for $\sinh^2 z$ (appearing in the denominator) of order
$2$ and that they are not zeros for
$\sinh\left(\mbox{{\small $\frac{m}{n}$}}\,z\right)$ and
$\sinh\left(\mbox{{\small $\frac{k}{n}$}}\,z\right)$
(due to $1<a=\frac{m}{n}<2$ and $0<b=\frac{k}{n}<1$). Thus, we conclude that
$z=z_1, z_{n-1}$ are double poles for $F(z)$.

We begin with computation of the residue $\mbox{Res}(F(z);z_1)$ at $z=z_1$. Thanks to
$\sinh z=-\sinh(z-i\pi)$ (or by direct computation) the power series expansion of
$\sinh z$ around $z_1=i\pi$ is given by
\begin{linenomath} \begin{align*}
\sinh z 
&=
-\bigl((z-z_1)+(z-z_1)^3/3!+(z-z_1)^5/5!+\cdots\bigr)\\
&=
-(z-z_1)\bigl(1+(z-z_1)^2/3!+(z-z_1)^4/5!+\cdots\bigr).
 \end{align*} \end{linenomath}
We thus get the following Laurent series expansion:
\begin{linenomath} \begin{align}\label{(5.4)} 
\frac{1}{\sinh^2z}
&=
\frac{1}{(z-z_1)^2} \cdot
\frac{1}{\bigl(1+(z-z_1)^2/3!+(z-z_1)^4/5!+\cdots\bigr)^2}
\notag\\
&=
\frac{1}{(z-z_1)^2} \cdot
\frac{1}{1+(z-z_1)^2/3+\mbox{higher even powers}}
\notag\\
&=
\frac{1}{(z-z_1)^2} 
\left(1-(z-z_1)^2/3+\mbox{higher even powers}\right).
 \end{align} \end{linenomath}
Since 
$$
\frac{d^{\ell}}{dz^{\ell}}
\,e^{isz}\Big|_{z=z_1}=(is)^{\ell}e^{isz}\big|_{z=z_1}=(is)^{\ell}e^{-\pi s},
$$
we have
\begin{equation}\label{(5.5)} 
e^{isz}=e^{-\pi s}
\left(1+is(z-z_1)-s^2(z-z_1)^2/2+\cdots\right).
\end{equation}
Computations
\begin{linenomath} \begin{align*}
&
\sinh\left(\mbox{{\small $\frac{m}{n}$}}\,z_1\right)
=i\sin\left(\mbox{{\small $\frac{m}{n}$}}\,\pi\right)
\ \bigl(=i\sin(at)\bigr),\\
&
\frac{d}{dz}\sinh\left(\mbox{{\small $\frac{m}{n}$}}\,z\right)\Big|_{z=z_1}
=\mbox{{\small $\frac{m}{n}$}}
\cosh\left(\mbox{{\small $\frac{m}{n}$}}\,z\right)\big|_{z=z_1}
=\mbox{{\small $\frac{m}{n}$}}\cos\left(\mbox{{\small $\frac{m}{n}$}}\,\pi\right)
\bigl(=a\cos(at)\bigr)
 \end{align*} \end{linenomath}
give rise to
\begin{equation}\label{(5.6)} 
\sinh\left(\mbox{{\small $\frac{m}{n}$}}\,z\right)
=i\sin(a\pi)+a\cos(a\pi)(z-z_1)+\cdots,
\end{equation}
and similarly
\begin{equation}\label{(5.7)} 
\sinh\left(\mbox{{\small $\frac{k}{n}$}}\,z\right)
=i\sin(b\pi)+b\cos(b\pi)(z-z_1)+\cdots.
\end{equation}
From (\ref{(5.4)})--(\ref{(5.7)})
the Laurent series expansion of $F(z)$ around $z=z_1$ is given by
\begin{linenomath} \begin{align}\label{(5.8)} 
&
\frac{e^{-\pi s}}{(z-z_1)^2}
\left(1-(z-z_1)^2/3+\mbox{higher even powers}\right)
\bigl(1+is(z-z_1)+\cdots\bigr)
\notag\\
&
\quad
\times 
\bigl(i\sin(a\pi)+a\cos(a\pi)(z-z_1)+\cdots\bigr)
\bigl(i\sin(b\pi)+b\cos(b\pi)(z-z_1)+\cdots\bigr). 
 \end{align} \end{linenomath}
The residue $\mbox{Res}(F(z);z_1)$ is nothing but the coefficient of $(z-z_1)^{-1}$ here,
i.e., that of $(z-z_1)$ in the product of the above four brackets (multiplied by
$e^{-\pi s}$). Since the starting term is $1$ and a $(z-z_1)$-term is absent in the first
bracket, what we have to compute is the coefficient of $(z-z_1)$ in the product of the
last three brackets. In this way we arrive at
\begin{linenomath} \begin{align*}
\mbox{Res}(F(z);z_1)
&=
e^{-\pi s}
\bigl(
i\sin(a\pi) \cdot b\cos(b\pi)+a\cos(a\pi) \cdot i\sin(b\pi)\\
& \hskip 7cm
+is \cdot i\sin(a\pi) \cdot i\sin(b\pi))\\
&=
ie^{-\pi s}
\bigl(
a\cos(a\pi)\sin(b\pi)+b\cos(b\pi)\sin(a\pi)
-s\sin(a\pi)\sin(b\pi)
\bigr).
 \end{align*} \end{linenomath}

We next move to computation of the residue $\mbox{Res}(F(z);z_{n-1})$ at
$z=z_{n-1} \ \bigl(=i(n-1)\pi\bigr)$.
Because of $\sinh z=-\sinh(z-(n-1)\pi i) \ \bigl(=-\sinh(z+\pi i) \bigr)$ with $n$ even
we have
\begin{linenomath} \begin{align*}
\sinh z 
&=
-\bigl((z-z_{n-1})+(z-z_{n-1})^3/3!+(z-z_{n-1})^5/5!+\cdots\bigr)\\
&=
-(z-z_{n-1})\bigl(1+(z-z_{n-1})^2/3!+(z-z_{n-1})^4/5!+\cdots\bigr)
 \end{align*} \end{linenomath}
(with the identical coefficients as in the power expansion around $z=z_1$), and hence
we have
$$
\frac{1}{\sinh^2z}
=\frac{1}{(z-z_{n-1})^2} 
\left(1-(z-z_{n-1})^2/3+\mbox{higher even powers}\right)
$$
again (see (\ref{(5.4)})).  Also, since
$\frac{d^{\ell}}{dz^{\ell}}\,e^{isz}\big|_{z=z_{n-1}}
=(is)^{\ell}e^{isz}\big|_{z=z_{n-1}}=(is)^{\ell}e^{-(n-1)\pi s}$,
(\ref{(5.5)}) has to be replaced by
$$
e^{isz}=e^{-(n-1)\pi s}\left(1+is(z-z_{n-1})-s^2(z-z_{n-1})^2/2+\cdots\right).
$$
So far we have not seen changes of coefficients except the obvious modification that
the factor $e^{-\pi s}$ in (\ref{(5.5)}) was replaced by $e^{-(n-1)\pi s}$.
On the other hand, since
\begin{linenomath} \begin{align*}
\sinh\left(\mbox{{\small $\frac{m}{n}$}}\,z_{n-1}\right)
&=i\sin\left(\mbox{{\small $\frac{m}{n}$}}\,(n-1)\pi\right)
=i\sin\left(m\pi-\mbox{{\small $\frac{m}{n}$}}\,\pi\right)\\
&=-i\sin\left(\mbox{{\small $\frac{m}{n}$}}\,\pi\right)
\ \bigl(=-i\sin(at)\bigr), \\
\frac{d}{dz}\sinh\left(\mbox{{\small $\frac{m}{n}$}}\,z\right)\big|_{z=z_{n-1}}
&=\mbox{{\small $\frac{m}{n}$}}
\cosh\left(\mbox{{\small $\frac{m}{n}$}}\,z\right)\big|_{z=z_{n-1}}
=\mbox{{\small $\frac{m}{n}$}}\cos\left(\mbox{{\small $\frac{m}{n}$}}\,(n-1)\pi\right)\\
&=\mbox{{\small $\frac{m}{n}$}}\cos\left(m\pi-\mbox{{\small $\frac{m}{n}$}}\,\pi\right)
=\mbox{{\small $\frac{m}{n}$}}\cos\left(\mbox{{\small $\frac{m}{n}$}}\,\pi\right)
\ \bigl(=a\cos(at)\bigr),
 \end{align*} \end{linenomath}
the power series expansions (\ref{(5.6)}) and (\ref{(5.7)}) are replaced by
\begin{linenomath} \begin{align*}
\sinh\left(\mbox{{\small $\frac{m}{n}$}}\,z\right)
&=
-i\sin(a\pi)+a\cos(a\pi)(z-z_{n-1})+\cdots,\\
\sinh\left(\mbox{{\small $\frac{k}{n}$}}\,z\right)
&=
-i\sin(b\pi)+b\cos(b\pi)(z-z_{n-1})+\cdots
 \end{align*} \end{linenomath}
with constant terms of the opposite sign. The four relevant expansions are now at our
disposal, and the same reasoning as before (see the product (\ref{(5.8)})) gives us the
following conclusion:
\begin{linenomath} \begin{align*}
\mbox{Res}(F(z);z_{n-1})
&=
e^{-(n-1)\pi s}
\bigl(
-i\sin(a\pi) \cdot b\cos(b\pi)
-a\cos(a\pi) \cdot i\sin(b\pi)\\
& \hskip 5cm
+is\,(-i\sin(a\pi))(-i\sin(b\pi))
\bigr)\\
&=
ie^{-(n-1)\pi s}
\bigl(
-a\cos(a\pi)\sin(b\pi)-b\cos(b\pi)\sin(a\pi)\\
& \hskip 6.8cm
-s\sin(a\pi)\sin(b\pi)
\bigr).
 \end{align*} \end{linenomath}

The sum (multiplied by $2\pi i$) of the two residues we have computed so far can be
rearranged in the following way:
\begin{linenomath} \begin{align*}
&
\hskip -0.2cm
2 \pi i 
\bigl(
\mbox{Res}(F(z);z_1)+\mbox{Res}(F(z);z_{n-1})
\bigr)\\
&
=2 \pi 
\bigl[
e^{-\pi s}
\bigl(
-a\cos(a\pi)\sin(b\pi)-b\cos(b\pi)\sin(a\pi)+s\sin(a\pi)\sin(b\pi)
\bigr)\\
& \qquad \qquad
+e^{-(n-1)\pi s}
\bigl(
a\cos(a\pi)\sin(b\pi)+b\cos(b\pi)\sin(a\pi)+s\sin(a\pi)\sin(b\pi)
\bigr)
\bigr]\\
&
=2\pi e^{-n\pi s/2}
\bigl[
e^{(n/2-1)\pi s}
\bigl(
-a\cos(a\pi)\sin(b\pi)-b\cos(b\pi)\sin(a\pi)+s\sin(a\pi)\sin(b\pi)
\bigr)\\
& \qquad \qquad
+e^{-(n/2-1)\pi s}
\bigl(
a\cos(a\pi)\sin(b\pi)+b\cos(b\pi)\sin(a\pi)+s\sin(a\pi)\sin(b\pi)
\bigr)
\bigr]\\
&
=4\pi e^{-n\pi s/2}
\bigl[
-\bigl(a\cos(a\pi)\sin(b\pi)+b\cos(b\pi)\sin(a\pi)\bigr)\sinh\left((n/2-1)\pi s\right)\\
& \hskip 7cm
+s\sin(a\pi)\sin(b\pi)\cosh\left((n/2-1)\pi s\right)
\bigl].
 \end{align*} \end{linenomath}
Therefore, by recalling (\ref{(5.2)}) and (\ref{(5.3)}) we conclude
\begin{linenomath} \begin{align}
&\frac{1}{2\pi}\int_{-\infty}^{\,\infty}f(t)\,e^{ist}\,dt \nn\\
&\quad=\frac{1}{\sinh(n\pi s/2)}\,\Big[
\sin(a\pi)\sin(b\pi) \cdot s\cosh\left((n/2-1)\pi s\right) \nn\\
&\hskip3.5cm
-\bigl(a\cos(a\pi)\sin(b\pi)+b\cos(b\pi)\sin(a\pi)\bigr)
\sinh\left((n/2-1)\pi s\right) \nn\\
&\hskip3.5cm+\mbox{lower order terms} \label{(5.9)}
\Big].
 \end{align} \end{linenomath}
A few remarks concerning ``lower order terms" are in order. Other candidates for poles
(inside of $\Gamma$) of $F(z)$ are
$$
z_{\ell}=i\ell \pi \quad (\mbox{for $\ell=2,3,\dots,n-2$}),
$$
where nature of singularities at these points (i.e., removable  singularities or poles of
order at most $2$) is determined 
according to values of $\sinh\left(\mbox{{\small $\frac{m}{n}$}}\,z_{\ell}\right)$,
$\sinh\left(\mbox{{\small $\frac{k}{n}$}}\,z_{\ell}\right)$
appearing in the numerator. Anyway, residues arising from them give us linear combinations 
of factors of the forms
$$
\sinh\left(\left(n/2-\ell'\right)\pi s\right), 
\quad \cosh\left(\left(n/2-\ell'\right)\pi s\right)
\quad \mbox{with $\ell'=2,3,\dots,n/2$}
$$
in the above big bracket \eqref{(5.9)} (possibly with the linear factor $s$ for double
poles). Indeed, the only source for exponential factors is  the power series expansions
of $e^{isz}$ around $z=z_{\ell}$ (see (\ref{(5.5)})), which actually gives rise to
$$
e^{isz_{\ell}}=e^{-\ell\pi s}=e^{-n\pi s/2}e^{(n/2-\ell)\pi s}
\quad (\ell=2,3,\dots,n-2).
$$
Thus, by recalling the factor $e^{-n\pi s/2}$ appearing in (\ref{(5.2)}), we get the
assertion.

The dominant term (as $s \to \pm\infty$) in the numerator of the Fourier transform is
$$
\sin(a\pi)\sin(b\pi) \cdot s\cosh\left((n/2-1)\pi s\right),
$$
and we observe
$$
\sin(a\pi)\sin(b\pi)<0
$$
thanks to $1 < a < 2$ and $0 < b < 1$ (see (\ref{(5.1)})). Consequently, the Fourier
transform takes negative values for $|s|$ large (i.e., failure of positive definiteness
for $f(s)$), and Lemma \ref{Lemma 5.2} has been proved.

\subsection{Fourier transform of $\left((\cosh(t/2)+\alpha)(\cosh t+\beta)\right)^{-1}$}  \label{sect:cosh}

Detailed information on positive definiteness for
$\left(\cosh^k(t/2)(\cosh t+\beta)^m\right)^{-1}$ will be needed to prove results on
geometric bridges in Example \ref{ex:geom.gen}. However, a direct computation for its
Fourier transform based on residue calculus seems hopeless due to the fact that poles of
higher orders have to be considered. Instead, in this subsection we compute the Fourier
transform in the special case $k=m=1$ with the additional parameter $\alpha$ as in the
theorem below (and then in Section \ref{sect:infdiv2} we will check higher order partial
derivatives relative to $\alpha$ and $\beta$ to achieve our goal).

\begin{thm}\label{Theorem 5.3}
For $\alpha \in (-1,1)$ and $\beta > 1$ we have
\begin{linenomath} \begin{align}\label{(5.10)} 
&
\frac{1}{4\pi}\int_{-\infty}^{\infty}
\frac{e^{ist}\,dt}{\bigl(\cosh(t/2)+\alpha\bigr)\bigl(\cosh t +\beta\bigr)}
\nonumber\\
& \qquad
=\frac{1}{\sinh(2\pi s)}
\Biggl[
\frac{\sinh(2\theta s)}{\sqrt{1-\alpha^2}\bigl(2\alpha^2-1+\beta\bigr)}
\nonumber\\
&
\hskip 3.5cm
-\,
\frac{
\sqrt{\frac{\beta-1}{2}} \cos(\lambda s) \sinh(\pi s)
-\alpha \sin(\lambda s) \cosh(\pi s)
}
{\sqrt{\beta^2-1}\bigl(\frac{\beta-1}{2}+\alpha^2\bigr)}
\Biggr],
 \end{align} \end{linenomath}
where $\theta=\cos^{-1}\alpha \in (0,\pi)$ and 
$\lambda=\log\left(\beta+\sqrt{\beta^2-1}\right)$, i.e., $\lambda>0$ is a solution of
$\cosh\lambda=\beta$.
\end{thm}
\begin{proof}
For a fixed $s \in \bR$ we set
$$
F(z)\equiv \frac{e^{isz}}{\bigl(\cosh(z/2)+\alpha\bigr)\bigl(\cosh z +\beta\bigr)}
$$
and compute its integral along the following rectangle $\Gamma$:
$$
\begin{array}{lll}
\Gamma_1&z=t,& t: -R \to R,\\
\Gamma_2&z=R+is,& s: 0 \to 4\pi,\\
\Gamma_3&z=t+4\pi i,& t: R \to -R,\\
\Gamma_4&z=-R+is,& s: 4\pi \to 0.
\end{array}
$$
Due to $\cosh((t+4\pi i)/2)=\cosh(t/2)$, $\cosh(t+4\pi i)=\cosh t$ and
$$
\lim_{R \to \infty}
\int_{\Gamma_2 \cup \Gamma_4} f(z)\,dz=0
$$
we have
\begin{linenomath} \begin{align*}
\lim_{R \to \infty}\int_{\Gamma}f(z)\,dz
&=\lim_{R \to \infty} \int_{\Gamma_1 \cup \Gamma_3}f(z)\,dz\\
&=\lim_{R \to \infty}
(1-e^{-4\pi s})
\int_{-R}^{R} \frac{e^{ist}ds}{\bigl(\cosh(t/2)+\alpha\bigr)\bigl(\cosh t +\beta\bigr)}\\
&=2e^{-2\pi s}\sinh(2\pi s)\int_{-\infty}^{\infty}
\frac{e^{ist}ds}{\bigl(\cosh(t/2)+\alpha\bigr)\bigl(\cosh t +\beta\bigr)},
 \end{align*} \end{linenomath}
and we will compute $\int_{\Gamma} f(z)\,dz$ by residue calculus.

It is easy to see that we have the following six simple poles inside of $\Gamma$:
$$
z_0=2i(\pi-\theta), \quad z_1=2i(\pi+\theta), \quad
\xi^{\pm}_0=i\pi\pm\lambda, \quad \xi^{\pm}_1=3i\pi\pm\lambda.
$$
When $\alpha=1$ (i.e., $\theta=0$), $2\pi i$ is a double pole. However, we assumed
$\alpha \in (-1,1)$ to avoid this complication. Note that the Fourier transform formula
(\ref{(5.10)}) itself remains valid for $\alpha=1$  by the obvious limiting argument with
the understanding
$$
\frac{\sinh(2\theta s)}{\sqrt{1-\alpha^2}}\bigg|_{\alpha=1}
=\lim_{\alpha \nearrow 1}\frac{\sinh(2\theta s)}{\sqrt{1-\alpha^2}}=2s
$$
(see \eqref{(5.17)} below).  We note
$$
\text{Res}(z_j,F(z))
=\frac{e^{isz_j}}{\frac{1}{2}\sinh(z_j/2)\bigl(\cosh z_j+\beta\bigr)},
\qquad j=0,1,
$$
and observe
\begin{linenomath} \begin{align*}
& e^{isz_0}=e^{-2(\pi-\theta)s}, \quad  e^{isz_1}=e^{-2(\pi+\theta)s},\\
&
\sinh(z_0/2)=i\sin(\pi-\theta)=i\sin\theta=i\sqrt{1-\alpha^2},\\
&
\sinh(z_1/2)=i\sin(\pi+\theta)=-i\sin\theta=-i\sqrt{1-\alpha^2},\\
&
\cosh(z_0)=\cos(2(\pi-\theta))=\cos(2\theta)=2\cos^2\theta-1=2\alpha^2-1,\\
&
\cosh(z_1)=\cos(2(\pi+\theta))=\cos(2\theta)=2\alpha^2-1.
 \end{align*} \end{linenomath}
Thus, we compute
\begin{linenomath} \begin{align*}
\text{Res}(z_0;F(z))&=
\frac{e^{-2(\pi-\theta)s}}{\frac{i}{2}\sqrt{1-\alpha^2}\bigl(2\alpha^2-1+\beta\bigr)}
\,=\,-\frac{2ie^{-2(\pi-\theta)s}}{\sqrt{1-\alpha^2}\bigl(2\alpha^2-1+\beta\bigr)},\\
\text{Res}(z_1;F(z))&=
\frac{e^{-2(\pi+\theta)s}}{-\,\frac{i}{2}\sqrt{1-\alpha^2}\bigl(2\alpha^2-1+\beta\bigr)}
\,=\,
\frac{2ie^{-2(\pi+\theta)s}}{\sqrt{1-\alpha^2}\bigl(2\alpha^2-1+\beta\bigr)},\\
 \end{align*} \end{linenomath}
and consequently we have
\begin{linenomath} \begin{align}\label{(5.11)} 
\text{Res}(z_0;F(z))+\text{Res}(z_1;F(z))
&=-\,\frac{2ie^{-2\pi s} \left(e^{2\theta s}-e^{-2\theta s}\right)}
{\sqrt{1-\alpha^2}\bigl(2\alpha^2-1+\beta\bigr)} \nonumber\\
&=-\,\frac{4ie^{-2\pi s} \sinh(2\theta s)}{\sqrt{1-\alpha^2}\bigl(2\alpha^2-1+\beta\bigr)}.
 \end{align} \end{linenomath}

We note
$$
\text{Res}(\xi^{\pm}_j;F(z))
=\frac{e^{is\xi^{\pm}_j}}{\bigl(\cosh(\xi^{\pm}_j/2)+\alpha\bigr)\sinh \xi^{\pm}_j},
\qquad j=0,1.
$$
We observe
\begin{linenomath} \begin{align*}
&
e^{is\xi^{\pm}_0}=e^{-\pi s \pm i\lambda s}, \quad
e^{is\xi^{\pm}_1}=e^{-3\pi s \pm i\lambda s},\\
&
\cosh(\xi^{\pm}_0/2)=\cosh((i\pi \pm \lambda)/2)=\pm i\sinh(\lambda/2)
=\pm i \sqrt{\frac{\cosh \lambda-1}{2}}=\pm i \sqrt{\frac{\beta-1}{2}},\\
&
\cosh(\xi^{\pm}_1/2)=\cosh((3i\pi \pm \lambda)/2)=\mp i\sinh(\lambda/2)
=\mp i \sqrt{\frac{\beta-1}{2}},\\
&
\sinh(\xi^{\pm}_0)=\sinh(i\pi\pm\lambda)=\mp\sinh\lambda=\mp\sqrt{\beta^2-1},\\ 
&
\sinh(\xi^{\pm}_1)=\sinh(3i\pi\pm\lambda)=\mp\sqrt{\beta^2-1}, 
 \end{align*} \end{linenomath}
and hence
\begin{linenomath} \begin{align*}
\text{Res}(\xi^{\pm}_0;F(z))
&=
\frac{e^{-\pi s \pm i\lambda s}}
{\Bigl(\pm i \sqrt{\frac{\beta-1}{2}}+\alpha\Bigr)\bigl(\mp \sqrt{\beta^2-1}\bigr)}
=
\frac{ie^{-\pi s \pm i\lambda s}}
{\sqrt{\beta^2-1}\Bigl(\sqrt{\frac{\beta-1}{2}}\mp i\alpha\Bigr)}\\
&=
\frac{ie^{-\pi s \pm i\lambda s}\Bigl( \sqrt{\frac{\beta-1}{2}} \pm i\alpha \Bigr)}
{\sqrt{\beta^2-1}\bigl(\frac{\beta-1}{2}+\alpha^2\bigr)},\\
\text{Res}(\xi^{\pm}_1;F(z))
&=
\frac{e^{-3\pi s \pm i\lambda s}}
{\Bigl(\mp i \sqrt{\frac{\beta-1}{2}}+\alpha\Bigr)\bigl(\mp \sqrt{\beta^2-1}\bigr)}
=
-\,\frac{ie^{-3\pi s \pm i\lambda s}}
{\sqrt{\beta^2-1}\Bigl(\sqrt{\frac{\beta-1}{2}}\pm i\alpha\Bigr)}\\
&=
-\,\frac{ie^{-3\pi s \pm i\lambda s}\Bigl( \sqrt{\frac{\beta-1}{2}} \mp i\alpha \Bigr)}
{\sqrt{\beta^2-1}\bigl(\frac{\beta-1}{2}+\alpha^2\bigr)}.
 \end{align*} \end{linenomath}
We compute
\begin{linenomath} \begin{align*}
\text{Res}(\xi^{+}_0;F(z))+\text{Res}(\xi^{-}_0;F(z))
&=
\frac{
ie^{-\pi s}
\left[
e^{i\lambda s}
\Bigl(\sqrt{\frac{\beta-1}{2}} + i\alpha \Bigr)
+
e^{-i\lambda s}
\Bigl(\sqrt{\frac{\beta-1}{2}} - i\alpha \Bigr)
\right]
}
{\sqrt{\beta^2-1}\bigl(\frac{\beta-1}{2}+\alpha^2\bigr)}\\
&=
\frac{
2ie^{-\pi s}
\left[
\sqrt{\frac{\beta-1}{2}} \cos(\lambda s)-\alpha \sin(\lambda s)
\right]
}
{\sqrt{\beta^2-1}\bigl(\frac{\beta-1}{2}+\alpha^2\bigr)},\\
\text{Res}(\xi^{+}_1;F(z))+\text{Res}(\xi^{-}_1;F(z))
&=
-\,
\frac{
ie^{-3\pi s}
\left[
e^{i\lambda s}
\Bigl(\sqrt{\frac{\beta-1}{2}} - i\alpha \Bigr)
+
e^{-i\lambda s}
\Bigl(\sqrt{\frac{\beta-1}{2}} + i\alpha \Bigr)
\right]
}
{\sqrt{\beta^2-1}\bigl(\frac{\beta-1}{2}+\alpha^2\bigr)}\\
&=
-\,
\frac{
2ie^{-3\pi s}
\left[
\sqrt{\frac{\beta-1}{2}} \cos(\lambda s)+\alpha \sin(\lambda s)
\right]
}
{\sqrt{\beta^2-1}\bigl(\frac{\beta-1}{2}+\alpha^2\bigr)}.
 \end{align*} \end{linenomath}
Both the quantities have $\sin$, $\cos$, and we conclude
\begin{linenomath} \begin{align}\label{(5.12)} 
&
\text{Res}(\xi^{+}_0;F(z))+\text{Res}(\xi^{-}_0;F(z))
+\text{Res}(\xi^{+}_1;F(z))+\text{Res}(\xi^{-}_1;F(z))
\nonumber\\
& \qquad
=
\frac{
2i
\left[
\sqrt{\frac{\beta-1}{2}} \cos(\lambda s) \bigl(e^{-\pi s}-e^{-3\pi s} \bigr)
-\alpha \sin(\lambda s) \bigl( e^{-\pi s}+e^{-3\pi s}\bigr)
\right]
}
{\sqrt{\beta^2-1}\Bigl(\frac{\beta-1}{2}+\alpha^2\Bigr)}
\nonumber\\
& \qquad
=
\frac{
4ie^{-2\pi s}
\left[
\sqrt{\frac{\beta-1}{2}} \cos(\lambda s) \sinh(\pi s)
-\alpha \sin(\lambda s) \cosh(\pi s)
\right]
}
{\sqrt{\beta^2-1}\Bigl(\frac{\beta-1}{2}+\alpha^2\Bigr)}.
 \end{align} \end{linenomath}
The desired Fourier transform formula (\ref{(5.10)}) is obtained as the sum of
(\ref{(5.11)}) and (\ref{(5.12)}) (multiplied by $2\pi i$).
\end{proof}

\subsection{Analysis of
$\left((\cosh(t/2)+\alpha)(\cosh t+\beta)\right)^{-1}$}  \label{sect:infdiv}

Here, we recall the Kolmogorov theorem (a version of L\'evy-Khintchine formula):
A function $f(t)$ on $\bR$ is the characteristic function of an infinitely divisible
probability measure with finite second moment if and only if there exist a finite positive
measure $\nu$ and a $\gamma\in\bR$ such that
$$
\log f(t)=i\gamma t+\int_{-\infty}^\infty\biggl({e^{its}-1-its\over s^2}\biggr)\,d\nu(s).
$$
Detailed accounts can be found in \cite{L,GK} for instance. We note that functions $f(t)$
we are dealing with here are all smooth and hence the ``finite second moment" condition
is automatic (see \cite[Section 2.3]{L} for instance).

In the following lemma we state two explicit examples of the Kolmogorov theorem obtained
in \cite[Lemma 2\,(ii) and Lemma 16]{Ko3} for later use and for the convenience of the
reader:

\begin{lemma}\label{Lemma 5.4} \
\begin{itemize}
\item[\rm(i)] For $a>0$ and $\theta\in[0,\pi)$,
$$
\log\biggl({1+\cos\theta\over\cosh(at)+\cos\theta}\biggr)
=\int_{-\infty}^\infty\left(e^{its}-1-ist\right)
{\cosh(\theta s/a)\over s\sinh(\pi s/a)}\,ds.
$$
\item[\rm(ii)] For $a>0$ and $\lambda\ge0$,
$$
\log\biggl({1+\cosh\lambda\over\cosh(at)+\cosh\lambda}\biggr)
=\int_{-\infty}^\infty\left(e^{its}-1-ist\right)
{\cos(\lambda s/a)\over s\sinh(\pi s/a)}\,ds.
$$
\end{itemize}
\end{lemma}

When $\alpha=0$ (i.e., $\theta=\pi/2$), (\ref{(5.10)}) reduces to
\begin{equation}\label{(5.13)} 
\int_{-\infty}^{\infty}
\frac{e^{ist}\,dt}{\cosh(t/2)\left(\cosh t+\beta\right)}
=
2\pi\,
\frac{
1-\sqrt{\frac{2}{\beta+1}}
\cos(\lambda s)
}
{(\beta-1)\cosh(\pi s)}
\quad (\geq 0),
\end{equation}
which corresponds to the special case $\alpha=0$ in the next result
(\cite[Theorem 4.13]{CM} and see also \cite[Section 7]{Ko3}).

\begin{cor}\label{Corollary 5.5}
We set
$$
G(t)\equiv 
\frac{1}{\bigl(\cosh(t/2)+\alpha\bigr)\bigl(\cosh t +\beta\bigr)}
$$
with $\alpha,\beta>-1$.
\begin{itemize}
\item[\rm(i)]  When $\beta>1$, $G(t)$ is positive definite if and only if
$\alpha \in (-1,0]$.
\item[\rm(ii)]  When $-1 <\beta \leq 1$, $G(t)$ is infinitely divisible for each
$\alpha \in (-1,\infty)$.
\end{itemize}
\end{cor}
\begin{proof}
Assume $\beta>1$. Due to (\ref{(5.13)}) $G(t)$ is positive definite for $\alpha=0$ and
remains so for $\alpha \in (-1,0]$ as well thanks to positive definiteness of
$$
\frac{\cosh(t/2)}{\cosh(t/2)+\alpha}=1+\frac{-\alpha}{\cosh(t/2)+\alpha}
$$
(see Lemma \ref{lemm:hyperb}\,(2)).
When $\alpha \in (0,1)$, we have $\theta=\cos^{-1}\alpha \in (0,\pi/2)$ in (\ref{(5.10)}).
Thus, the dominant terms in the big bracket in the right side of (\ref{(5.10)}) are
$\cos(\lambda s)\sinh(\pi s)$ and $\sin(\lambda s)\cosh(\pi s)$ so that the quantity
in the big bracket takes both positive and negative values for $|s|$ large. 
To prove (i), it remains to show failure of positive definiteness for $\alpha \geq 1$.
However this follows from positive definiteness of
$$
\frac{\cosh(t/2)+\alpha}{\cosh(t/2) + \half}=1+\frac{\alpha-\half}{\cosh(t/2) + \half}
$$
for instance (and the already known failure of positive definiteness for $\alpha=\half$).

Next, assume $-1<\beta\le1$. The statement (ii) is obvious for $\alpha \in (-1,1]$, $G(t)$
being the product of two infinitely divisible functions under these circumstances. 
On the other hand, when $\alpha>1$, we have
$$
\log \bigl(
(1+\alpha)(1+\beta)G(t)
\bigr)
=\int_{-\infty}^{\infty}\left(
e^{ist}-1-ist
\right)
\left(\frac{\cos(2\theta s)}{s\sinh(2\pi s)}+\frac{\cosh(\lambda s)}
{s\sinh(\pi s)}\right)ds
$$
with $\theta=\log\left(\alpha+\sqrt{\alpha^2-1} \right)$ and $\lambda=\cos^{-1}\beta$
(by Lemma \ref{Lemma 5.4}). The density here can be written as
\begin{equation}\label{(5.14)} 
\frac{\cos(2\theta s)+2\cosh(\pi s)\cosh(\lambda s)}
{s\sinh(2\pi s)},
\end{equation}
which is certainly positive.
\end{proof}

\subsection{Analysis of
$\left(\cosh^k(t/2)(\cosh t+\beta)^m\right)^{-1}$}  \label{sect:infdiv2}

In this subsection we obtain a result which will be used in Theorem \ref{Theorem 6.6} of 
Section~\ref{sect:geombrdgpf} to obtain a bound for the interval in which
$g_{1,\nu,\lambda}$ is in $\cK^+$.

We assume $\beta > 1$ and $\alpha \in (-1,0]$ as in Corollary \ref{Corollary 5.5}\,(i). 
Under these circumstances the density (\ref{(5.14)}) is switched to
$$
\log\bigl((1+\alpha)(1+\beta)G(t)\bigr)
=\int_{-\infty}^\infty\left(e^{its}-1-ist\right)
{\cosh(2\theta s)+2\cos(\lambda s)\cosh(\pi s)\over s\sinh(2\pi s)}\,ds
$$
with $\theta=\cos^{-1}\alpha \in [\pi/2,\pi)$ and
$\lambda=\log\left(\beta+\sqrt{\beta^2-1}\right)$. 
Thus, the positive definite function $G(t)$ (Corollary \ref{Corollary 5.5}\,(i))
is infinitely divisible if and only if
$$
\cosh(2\theta s)+2\cos(\lambda s)\cosh(\pi s) \geq 0, \qquad s \in \bR.
$$ 
This is quite a delicate condition, but for the extreme value $\alpha=0$ 
(i.e., $\theta=\pi/2$) the condition simply means
$$
\left(1+2\cos(\lambda s)\right)\cosh(\pi s) \geq 0,
$$
and it is never fulfilled for any $\beta>1$ (which is exactly \cite[Theorem 15]{Ko3}).

We take higher order partial derivatives $\partial_{\beta}^{m-1}\partial_{\alpha}^{k-1}$
from the Fourier transform formula (\ref{(5.10)}) (with the variable $s$ fixed). It is
obvious that from the left side we get a scalar multiple of
$$
\int_{-\infty}^{\infty}
\frac{e^{ist}\,ds}{\bigl(\cosh(t/2)+\alpha\bigr)^k\bigl(\cosh t +\beta\bigr)^m},
$$
and for the special value $\alpha = 0$ the above integral reduces to
\begin{equation}\label{(5.15)} 
\int_{-\infty}^{\infty}
\frac{e^{ist}\,ds}{\cosh^k(t/2)\bigl(\cosh t +\beta\bigr)^m}.
\end{equation}
Therefore, behavior on the Fourier transform (\ref{(5.15)}) can be seen by computing
$\partial_{\beta}^{m-1}\partial_{\alpha}^{k-1}$ of the right side of (\ref{(5.10)})
at first and then by substituting $\alpha=0$.

The right side $R(\alpha,\beta)$ of the formula (\ref{(5.10)}) consists of the three terms:
\begin{linenomath} \begin{align}\label{(5.16)} 
R(\alpha,\beta)
&=
F_0(\alpha,\beta) \,\frac{\sinh(2\theta s)}{\sinh(2\pi s)}
-F_{c}(\alpha,\beta)\cos(\lambda s) \,\frac{\sinh(\pi s)}{\sinh(2\pi s)}
+F_{s}(\alpha,\beta)\sin(\lambda s) \,\frac{\cosh(\pi s)}{\sinh(2\pi s)}  \qquad 
\nn
\\  ~~ \nn \\
&=
F_0(\alpha,\beta) \,\frac{\sinh(2\theta s)}{\sinh(2\pi s)}
-
F_{c}(\alpha,\beta) \,\frac{\cos(\lambda s)}{2\cosh(\pi s)}
+F_{s}(\alpha,\beta) \,\frac{\sin(\lambda s)}{2\sinh(\pi s)}
 \end{align} \end{linenomath}
with
\begin{linenomath} \begin{align*}
F_0(\alpha,\beta)
&=
\frac{1}{\sqrt{1-\alpha^2}\bigl(2\alpha^2-1+\beta\bigr)},\\
F_c(\alpha,\beta)
&=
\frac{
\sqrt{\frac{\beta-1}{2}}}{\sqrt{\beta^2-1}\bigl(\frac{\beta-1}{2}+\alpha^2\bigr)},\\
F_s(\alpha,\beta)
&=
\frac{\alpha}{\sqrt{\beta^2-1}\bigl(\frac{\beta-1}{2}+\alpha^2\bigr)}.
 \end{align*} \end{linenomath}
We note
\begin{equation}\label{(5.17)} 
{d\theta\over d\alpha}=-\,\frac{1}{\sin \theta}=-\,\frac{1}{\sqrt{1-\alpha^2}},
\qquad
{d\lambda\over d\beta}=\frac{1}{\sqrt{\beta^2-1}},
\end{equation}
and consequently
\begin{linenomath} \begin{align*}
&
\partial_{\alpha} \sinh(2\theta s)=-\,\frac{2s\cosh(2\theta s)}{\sqrt{1-\alpha^2}}, \qquad
\partial_{\alpha} \cosh(2\theta s)=-\,\frac{2s\sinh(2\theta s)}{\sqrt{1-\alpha^2}},\\
&
\partial_{\beta} \sin(\lambda s)=\frac{s\cos(\lambda s)}{\sqrt{\beta^2-1}}, \qquad
\partial_{\beta} \cos(\lambda s)=-\,\frac{s\sin(\lambda s)}{\sqrt{\beta^2-1}}.
 \end{align*} \end{linenomath}

We begin with the first term in (\ref{(5.16)}). Since  $\sinh(2\theta s)$ and
$\cosh(2\theta s)$ behave like  constants against $\partial_{\beta}$,
$\partial_{\beta}^{m-1}\partial_{\alpha}^{k-1}$ of the first term is a polynomial of $s$
of degree at most $k-1$ with coefficients $\sinh(2\theta s)$,
$\cosh(2\theta s)$, $1/\sinh(2\pi s)$ and so on. Therefore, the substitution $\alpha=0$
(i.e., $\theta=\pi/2$) gives rise to a polynomial of $s$ of degree at most $k-1$ with
coefficients containing
$$
\frac{\sinh(2\theta s)}{\sinh(2\pi s)}\bigg|_{\theta=\pi/2}=\frac{1}{2\cosh(\pi s)},
\qquad 
\frac{\cosh(2\theta s)}{\sinh(2\pi s)}\bigg|_{\theta=\pi/2}=\frac{1}{2\sinh(\pi s)}. 
$$

The same procedure for the second and third terms in (\ref{(5.16)}) obviously gives rise
to a polynomial of $s$ of degree at most $m-1$. It is important to make sure that
the order is exactly $m-1$, and we will closely check the coefficient of $s^{m-1}$. 
For this purpose we begin with the third term in (\ref{(5.16)}) and we note
\begin{linenomath} \begin{align*}
&F_s(\alpha,\beta)
=
\frac{1}{2\sqrt{\beta^2-1}}
\left(\frac{1}{\alpha+i\sqrt{\frac{\beta-1}{2}}}
+
\frac{1}{\alpha-i\sqrt{\frac{\beta-1}{2}}}
\right),\\
&
\partial_{\alpha}^{k-1} F_{s}(\alpha,\beta)
=
\frac{(-1)^{k-1} (k-1)!}{2\sqrt{\beta^2-1}}
\left(
\frac{1}{\left(\alpha+i\sqrt{\frac{\beta-1}{2}}\right)^{k}}
+\frac{1}{\left(\alpha-i\sqrt{\frac{\beta-1}{2}}\right)^{k}}
\right),\\
&
\partial_{\alpha}^{k-1}
\left(F_{s}(\alpha,\beta)\cdot\frac{\sin(\lambda s)}{2\sinh(\pi s)}
\right) 
=
\partial_{\alpha}^{k-1} F_{s}(\alpha,\beta)
\,
\frac{\sin(\lambda s)}{2\sinh(\pi s)}.
 \end{align*} \end{linenomath}
So far no $s$-terms show up because $\lambda$ just depends on $\beta$. We then take
derivatives relative to $\beta$. It is plain to see that the highest $s^{m-1}$-term
arises from
$$
\partial_{\alpha}^{k-1} F_{s}(\alpha,\beta)
\,
\partial_{\beta}^{m-1}
\left(\frac{\sin(\lambda s)}{2\sinh(\pi s)}\right)
=\partial_{\alpha}^{k-1} F_{s}(\alpha,\beta)
\,
\frac{\partial_{\beta}^{m-1} \sin(\lambda s)}{2\sinh(\pi s)}.
$$
From (\ref{(5.17)}) we also easily observe
\begin{equation}\label{(5.18)} 
\hskip 0.2cm
\partial_{\beta}^{m-1} \sin(\lambda s)
=
\begin{cases}
{\displaystyle \pm\frac{s^{m-1}}{(\beta^2-1)^{(m-1)/2}}
\,\sin(\lambda s)+\mbox{lower $s$-terms}}
& \text{(for $m$ odd)},\\
{\displaystyle \pm\frac{s^{m-1}}{(\beta^2-1)^{(m-1)/2}}
\,\cos(\lambda s)+\mbox{lower $s$-terms}}
& \text{(for $m$ even)}.
\end{cases}
\end{equation}
By substituting $\alpha=0$, we observe
$$
\partial_{\alpha}^{k-1} F_{s}(\alpha,\beta)
\Big|_{\alpha=0}
=
\frac{(-1)^{k-1} (k-1)!}{2\sqrt{\beta^2-1}}
\left(
\frac{1}{\left(i\sqrt{\frac{\beta-1}{2}}\right)^{k}}
+\frac{1}{\left(-i\sqrt{\frac{\beta-1}{2}}\right)^{k}}
\right)
\neq 0
$$
as long as $k$ is even.  From the discussion so far, for $k$ even the highest
$s^{m-1}$-term arising from
$$
\partial_{\beta}^{m-1}
\partial_{\alpha}^{k-1}
\left(F_{s}(\alpha,\beta)\cdot\frac{\sin(\lambda s)}{2\sinh(\pi s)}
\right)\bigg|_{\alpha=0} 
$$
is a non-zero scalar (of course depending upon $\beta$) multiple of
$$
\frac{s^{m-1}\sin(\lambda s)}{\sinh(\pi s)} \quad (\mbox{for $m$ odd})
\quad
\mbox{or}
\quad
\frac{s^{m-1}\cos(\lambda s)}{\sinh(\pi s)} \quad (\mbox{for $m$ even})
$$
depending upon the parity of $m$.

We next move to the second term in (\ref{(5.16)}). We note
\begin{linenomath} \begin{align*}
&F_c(\alpha,\beta)
=
\frac{1}{2i\sqrt{\beta^2-1}}
\left(\frac{1}{\alpha+i\sqrt{\frac{\beta-1}{2}}}
-
\frac{1}{\alpha-i\sqrt{\frac{\beta-1}{2}}}
\right),\\
&\partial_{\alpha}^{k-1} F_{c}(\alpha,\beta)
=
\frac{(-1)^{k-1} (k-1)!}{2i\sqrt{\beta^2-1}}
\left(
\frac{1}{\left(\alpha+i\sqrt{\frac{\beta-1}{2}}\right)^{k}}
-\frac{1}{\left(\alpha-i\sqrt{\frac{\beta-1}{2}}\right)^{k}}
\right).
 \end{align*} \end{linenomath}
The presence of the minus sign this time in the big bracket enables us
to conclude
\bee
\partial_{\alpha}^{k-1} F_{c}(\alpha,\beta)\Big|_{\alpha=0} \neq 0
\eee
for $k$ odd.  Since the formula akin to (\ref{(5.18)}) is available to $\cos(\lambda s)$,
for $k$ odd the highest $s^{m-1}$-term arising from
\bee
\partial_{\beta}^{m-1}
\partial_{\alpha}^{k-1}
\left(F_{c}(\alpha,\beta)\,\frac{\cos(\lambda s)}{2\cosh(\pi s)}
\right)\bigg|_{\alpha=0} 
\eee
is a non-zero scalar multiple of
\bee
\frac{s^{m-1}\cos(\lambda s)}{\cosh(\pi s)} \quad (\mbox{for $m$ odd})
\quad
\mbox{or}
\quad
\frac{s^{m-1}\sin(\lambda s)}{\cosh(\pi s)} \quad (\mbox{for $m$ even})
\eee
this time.

Summing up the discussions so far, we conclude: For $s$ large the leading terms of\break
$\partial_\beta^{m-1}\partial_\alpha^{k-1}R(\alpha,\beta)$ (which arise form the first
term and the last two terms  in (\ref{(5.16)}) respectively) are  non-zero scalar
multiples of
$$
\frac{s^{k-1}}{e^{\pi s}} \quad \mbox{and} \quad
\frac{s^{m-1}\sin(\lambda s+\delta)}{e^{\pi s}}
$$
regardless of the parity of $k$.

We are now ready to prove the following result:
\begin{thm}\label{Theorem 5.6}
We assume $\beta>1$ and set
$$
H(t)\equiv \frac{1}{\cosh^k(t/2)\bigl(\cosh t+\beta\bigr)^m}
$$
for positive integers $k,m$.
\begin{itemize}
\item[\rm(i)] $H(t)$ is positive definite if and only if $k \geq m$.
\item[\rm(ii)] $H(t)$ is infinitely divisible if and only if $k \geq 2m$.
\end{itemize}
\end{thm}
\begin{proof} We set $\lambda=\log\left(\beta+\sqrt{\beta^2-1} \right)>0$ as in
Theorem \ref{Theorem 5.3}.

(i) \ Firstly we assume $k \geq m$. Since
$$
\frac{1}{\cosh^k(t/2)\bigl(\cosh t+\beta\bigr)^m}=
\frac{1}{\cosh^{k-m}(t/2)}
\left(\frac{1}{\cosh(t/2)\bigl(\cosh t+\beta\bigr)}\right)^m,
$$
positive definiteness of $H(t)$ follows from Corollary \ref{Corollary 5.5}\,(i)
(or rather the paragraph before the corollary). On the other hand, when $k<m$, for $s$
large the dominant term in (\ref{(5.15)}) (i.e., the Fourier transform of $H(t)$) is
$$
\frac{s^{m-1}\sin(\lambda s+\delta)}{e^{\pi s}}
$$
as was mentioned in the paragraph right before the theorem. Thus, the Fourier transform
admits both positive and negative values and hence $H(t)$ cannot be positive definite.

(ii) \ By Lemma \ref{Lemma 5.4} we have
$$
\log\bigl((1+\beta)^mH(t)\bigr)
=\int_{-\infty}^{\infty}\left(e^{ist}-1-ist\right)F(s)\,ds
$$
with
$$
F(s)=\frac{k}{2s\sinh(\pi s)}+\frac{m\cos(\lambda s)}{s\sinh(\pi s)}
=\frac{k+2m\cos(\lambda s)}{2s\sinh(\pi s)}.
$$
Thus, $H(t)$ is infinitely divisible if and only if $k+2m\cos(\lambda s) \geq 0$, i.e.,
$k \geq 2m$.
\end{proof}

Note that the optimal case in (ii) (i.e., case $k=2m$) corresponds to infinite divisibility
of $\bigl((\cosh t+1)(\cosh t+\beta)\bigr)^{-1}$ (see \cite[\S 7.1]{Ko3}) because of
$\cosh^2(t/2)=(\cosh t+1)/2$. The function $1/\cosh(t/2)$ is positive definite (and
indeed infinitely divisible) while $1/(\cosh t+\beta)$ (with $\beta >1$) is not (see
Lemma \ref{lemm:hyperb}\,(2)). Thus, intuition might suggest that as far as the function
$H(t)$ is concerned one has higher (resp., lower) chance for positive definiteness and/or
infinite  divisibility as $k$ (resp., $m$) increases. The theorem completely clarifies
where proper balance is taken.

\section{Proofs of results in Section~\ref{sect:examps}}  \label{sect:exampfs}

\subsection{Results on WYD family}  \label{sect:WYD}

For the functions $k_p^{\WYD}$, $p \in [-1,2]$, defined by \eqref{WYD}, we will prove the
next results stated in Example \ref{ex:WYD}.

\begin{thm} \label{Theorem 6.1} \
\begin{itemize}
\item[\rm(a)] The function $k_p^\WYD$ belongs to $\cK^+$ if and only if $p \in [0,1]$,
\item[\rm(b)] The function $k_p^\WYD$ belongs to $\cK^-$ if and only if
$p \in \big[-1,-\half \big] \cup \big[\tfrac{3}{2}, 2 \big]$.
\end{itemize}
\end{thm}

\begin{proof}
We may and do assume $p \in \bigl[\half,2\bigr]$ in view of the symmetry
$k_p^\WYD=k_{1-p}^\WYD$, and set
\begin{linenomath} \begin{align*}
&
f_p(t)\equiv e^{t}k_p^\WYD(e^{2t})=\frac{1}{p(1-p)}\cdot
\frac{\sinh (pt) \sinh ((1-p)t)}{\sinh^2 t},\\
&
g_p(t)\equiv e^{t}/k_p^\WYD(e^{-2t})=p(1-p)\cdot
\frac{\sinh^2 t}{\sinh (pt) \sinh ((1-p)t)}
\notag
 \end{align*} \end{linenomath}
for convenience. Theorem \ref{thm:CP2} says that we have to determine when $f_p(t)$ and
$g_p(t)$ are positive definite.

Positive definiteness for $f_p(t)$ with $p \in \bigl[\half,1\bigr]$ is well-known (where
$f_1(t)$ is understood  as $t/\sinh t$). When $p \in (1,2]$, we have
$$
f_p(t)=\frac{1}{p(p-1)}\cdot\frac{\sinh (pt) \sinh ((p-1)t)}{\sinh^2 t},
$$
which fails to be positive definite thanks to Theorem \ref{Theorem 5.1} (due to the
presence of $\sinh (pt)$ with $p>1$ in the numerator). Thus, we have shown that $f_p(t)$
is positive definite if and only if $p \in \bigl[\half,1\bigr]$, that is, (a) is proved.

To prove (b), we next check positive definiteness for $g_p(t)$. When
$p \in \bigl[\half,1\bigr]$, we have
$$
\lim_{t \to \pm\infty}g_p(t)=+\infty
$$
(where $g_1(t)$ is understood as $\sinh t/t$) so that $g_p(t)$ cannot be positive definite. 
We now move to the case $p \in (1,2]$ so that we use the expression
$$
g_p(t)=p(p-1)\,\frac{\sinh^2 t}{\sinh (pt) \sinh ((p-1)t)}.
$$
When $2>p+(p-1)$ (i.e., $p < {3\over2}$), $g_p(t)$ once again diverges as $t \to \pm\infty$
and fails to be positive definite. Thus, it remains to show positive definiteness for
$p \geq {3\over2}$. For the extreme value $p={3\over2}$ we compute
\begin{linenomath} \begin{align*}
g_{3/2}(t)
&=
\frac{3}{4}\cdot\frac{\sinh^2 t}{\sinh(3t/2)\sinh(t/2)}
=\frac{3}{4}\cdot\frac{\sinh^2 t}{\sinh^2(t/2)\left(4\cosh^2(t/2)-1\right)}\\
&=
\frac{3\cosh^2(t/2)}{4\cosh^2(t/2)-1}
=
\frac{3\left(\cosh t+1\right)/2}{2\left(\cosh t+1\right)-1}
=\frac{3}{4}
\left(1+\frac{\half}{\cosh t+\half}
\right).
 \end{align*} \end{linenomath}
Since $1/\left(\cosh t + \half\right)$ is positive definite (see
Lemma \ref{lemm:hyperb}\,(2)), so is $g_{3/2}(t)$. Finally, the obvious identity
$$
\frac{\sinh^2 t}{\sinh (pt) \sinh ((p-1)t)}
=
\frac{\sinh^2 t}{\sinh(3t/2)\sin(t/2)}
\cdot
\frac{\sinh(3t/2)\sin(t/2)}{\sinh(pt)\sin((p-1)t)}
$$
takes care of the remaining case (i.e., $p \in \bigl({3\over2},2\bigr]$).
\end{proof}

\subsection{Proofs for Example \ref{Example 4.3}}   \label{sect:convextpf}

We now prove the claims in  Example \ref{Example 4.3}.

\begin{thm}\label{thm:newext}
We assume $\nu \in (0,1)$ and $\lambda\in[0,1]$. Then, the function 
$a_{1,\nu ,\lambda}$ defined in \eqref{afunc} belongs to ${\mathcal K}^+$ if and only if 
\be  \label{nucond}
\lambda \leq \frac{2\sqrt{\nu}}{ \left(1+\sqrt{\nu}\right)^2}
 = \frac{2}{ \big(\nu^{1/4} + \nu^{-1/4} \big)^2} .
\ee
Moreover,  $a_{1,\nu ,\lambda}$ is an extreme point in ${\mathcal K}^+$ if and only if 
equality holds in \eqref{nucond}.
\end{thm}

\begin{proof}
With $\beta=(1+\nu^2)/2\nu \ (>1)$ we compute
\begin{linenomath} \begin{align*}
e^{t/2}a_{1,\nu,\lambda}(e^t)
&=
\lambda(\beta+1)\,\frac{\cosh(t/2)}{\cosh t+\beta}+(1-\lambda)\,\frac{1}{\cosh(t/2)}\\
&=
\frac{\lambda(\beta+1)\left(\cosh t+1\right)/2+(1-\lambda)\left(\cosh t+\beta\right)}
     {\cosh(t/2)\left(\cosh t+\beta\right)}\\
&=
\frac{\bigl(\lambda(\beta-1)/2+1\bigr)\left(\cosh t+\beta\right)-\lambda(\beta^2-1)/2}
     {\cosh(t/2)\left(\cosh t+\beta\right)}\\
&=
\frac{\lambda(\beta-1)/2+1}{\cosh(t/2)}
-\,
\frac{\lambda(\beta^2-1)/2}{\cosh(t/2)\left(\cosh t+\beta\right)}.
 \end{align*} \end{linenomath}
Let us recall
(\ref{(5.13)})
(where the symbol $\lambda$ for $\cosh^{-1}\beta$ there is changed to $\alpha$ to avoid
the obvious confusion) and
$$
\int_{-\infty}^{\infty}\frac{e^{ist}\,dt}{\cosh(t/2)}
=\frac{2\pi}{\cosh(\pi s)}. 
$$
The Fourier transform is thus given by
\begin{linenomath} \begin{align*}
\frac{1}{\pi}\int_{-\infty}^{\infty}
e^{t/2}a_{1,\nu,\lambda}(e^t)\,e^{ist}\,dt
&=
\frac{\lambda(\beta-1)+2}{\cosh(\pi s)}
-\,
\frac{\lambda(\beta+1)\left(1-\sqrt{\frac{2}{\beta+1}}\cos(\alpha s)\right)}{\cosh(\pi s)}\\
&=
\frac{N(s)}{\cosh(\pi s)}
 \end{align*} \end{linenomath}
with the numerator
\begin{linenomath} \begin{align*}
N(s)
&\equiv 2(1-\lambda)+\lambda\sqrt{2\left(\beta+1\right)}\cos(\alpha s).
 \end{align*} \end{linenomath}
This computation says that $e^{t/2}a_{1,\nu,\lambda}(e^t)$ is positive definite 
(i.e., $a_{1,\nu,\lambda} \in {\mathcal K}^+$) if and only if
$$
\lambda\sqrt{2\left(\beta+1\right)} \leq 2(1-\lambda).
$$
It is obviously satisfied for $\lambda=0$ (which corresponds  to the obvious positive
definiteness of $e^{t/2}k_1^{\ext}(e^t)$) while for $\lambda >0$ the requirement is
the same as
$$
\sqrt{\frac{1+\beta}{2}} \leq \frac{1-\lambda}{\lambda}
\ \Longleftrightarrow \
\lambda \leq \left(\sqrt{\frac{1+\beta}{2}}+1\right)^{-1}.
$$
Finally, we compute 
$$
\sqrt{\frac{1+\beta}{2}}+1
=
\sqrt{\frac{(1+\nu)^2}{4\nu}}+1
=\frac{\left(1+\sqrt{\nu}\right)^2}{2\sqrt{\nu}},
$$
which proves the first part.

Now, we set $\lambda(\nu)\equiv2\sqrt\nu/(1+\sqrt\nu)^2$ and prove the second
part. It is obvious that $a_{1,\nu ,\lambda}$ is not an extreme point of $\cK^+$ if
$\lambda<\lambda(\nu)$. To prove the converse we assume that
$a_{1,\nu ,\lambda(\nu)}(x)=\lambda k_1(x)+(1-\lambda)k_2(x)$ with some
$\lambda \in (0,1)$, $k_1,k_2 \in {\mathcal K}^+$ and hence
$$
k_i(x)=\int_{[0,1]}k_{\nu}^{ext}(x)\,dm_i(\nu), \qquad i=1,2,
$$
with representing probability measures $m_i$ on $[0,1]$. From the uniqueness of a
representing measure we have
$$
\lambda m_1 + (1-\lambda) m_2=(1-\lambda (\nu))\delta_1+\lambda(\nu)\delta_{\nu }. 
$$
In particular, we have
\begin{equation}\label{(6.2)} 
\left\{
\begin{array}{l}
\lambda m_1(\{1\})+(1-\lambda)m_2(\{1\})=1-\lambda (\nu ),\\
\lambda m_1(\{\nu \})+(1-\lambda)m_2(\{\nu \})=\lambda (\nu ),
\end{array}
\right.
\end{equation}
and they sum up to
\begin{equation}\label{(6.3)} 
\lambda \bigl(m_1(\{1\})+m_1(\{\nu \}\bigr)
+(1-\lambda)\bigl(m_2(\{1\})+m_2(\{\nu \})\bigr)=1.
\end{equation}
On the other hand, we have $m_i(\{1\})+m_i(\{\nu \}) \leq 1$ ($i=1,2$) because $m_i$'s are
probability measures. Therefore, (\ref{(6.3)}) guarantees $m_i(\{1\})+m_i(\{\nu \}) = 1$
($i=1,2$). Hence, both of $m_1, m_2$ are supported on the two-point set $\{\nu ,1\}$,
that is, of the form 
$$
m_i=(1-a_i)\delta_1+a_i\delta_{\nu } \quad (i=1,2)
$$
with $a_i=m_i(\{\nu \}) \in [0,1]$. 
Since $k_i \in {\mathcal K}^+$ ($i=1,2$), the first part of the theorem implies
$a_i \leq \lambda (\nu)$ and hence the second equation of (\ref{(6.2)}) forces
$a_1=a_2=\lambda (\nu)$, i.e., $k_1=k_2=a_{1,\nu ,\lambda (\nu )}$.
\end{proof}

Since $2\sqrt{\nu}/\left(1+\sqrt{\nu}\right)^2 < \half$  (with $\nu \neq 1$), we have
$a_{1,\nu,\lambda} \not\in {\mathcal K}^+$  as long as $\lambda \geq \half$ (and
$\nu \in [0,1)$).  Choose the extreme value
$\lambda=\lambda(\nu)=\left(\sqrt{(\beta+1)/2}+1\right)^{-1}$ with $\beta=(1+\nu^2)/2\nu$.
Then, it is straightforward to compute
$$
e^{t/2}a_{1,\nu,\lambda_0}(e^t)
=
\sqrt{\frac{\beta+1}{2}} \cdot \frac{\cosh t +\sqrt{2(\beta+1)}-1}
                                    {\cosh(t/2)\left(\cosh t+\beta \right)}.
$$
It is not clear if this positive definite function is infinitely divisible.

\subsection{Results on geometric bridges}  \label{sect:geombrdgpf}

We consider  the geometric bridges $g_{1,\nu,\lambda}(x)$, $0 \leq \lambda \leq 1$, 
between $k_1^{\ext}$ and $k_{\nu}^{\ext}$ with $\nu \in [0,1)$ given by \eqref{(4.7)}
and described  in Example \ref{ex:geom.gen}.
Since the case $\nu=0$ was settled in Example \ref{ex:geombrdg}:, we assume
$\nu \in (0,1)$. Our main result is
\begin{thm}  \label{thm:geombd}
For each fixed $\nu \in (0,1)$, there is a critical $\lambda_c$
{\rm(}dependent on $\nu${\rm)} such that the function $g_{1,\nu,\lambda}$ is in $\cK^+$
for $\lambda \in \big[0,\lambda_c \big]$. Moreover,
\be  \label{geom.crit}
\frac{1}{4} \leq \lambda_c(\nu) \leq \frac{1}{3} \qquad
\mbox{for each}\quad \nu \in (0,1).
\ee
\end{thm}
This will follow from a series of lemmas and theorems below, which are of independent
interest.  First, observe that one can rewrite \eqref{(4.7)} as
\begin{linenomath} \begin{align*}
g_{1,\nu,\lambda}(x)
& =x^{-1/2}\biggl({2\over x^{1/2}+x^{-1/2}}\biggr)^{1-2\lambda}
\left({(1+\nu)^2\over2\nu
\left({x+x^{-1}\over2}+{1+\nu^2\over2\nu}\right)
}\right)^\lambda.  \\
\intertext{and define for $\beta=\frac{1+\nu^2}{2\nu} >1 $,} 
f_{\nu,\lambda}(t) & \equiv e^{t/2}g_{1,\nu,\lambda}(e^t)
= \frac{1}{\cosh^{1-2\lambda}(t/2)}
\left(\frac{1+\beta}{\cosh t +\beta}\right)^{\lambda}.
 \end{align*} \end{linenomath} 
Recall that Theorem \ref{thm:CP2} implies that $g_{1,\nu,\lambda} \in \cK^+$ if and
only if  $f_{\nu,\lambda}(t)$ is positive definite.   
\begin{lemma}\label{Lemma 6.4}
The function $f_{\nu,\lambda}(t)$ is infinitely divisible if and only if
$0 \leq \lambda \leq {1\over4}$.
\end{lemma}
\begin{proof}
With $\alpha=\log\left(\beta+\sqrt{\beta^2-1}\right) > 0$ by Lemma \ref{Lemma 5.4} we have
\begin{linenomath} \begin{align*}
\log f_{\nu,\lambda}(t)
&=\lambda\log\left(\frac{1+\beta}{\cosh t +\beta}\right)
+(1-2\lambda)\log \left(\frac{1}{\cosh(t/2)}\right)\\
&=
\int_{-\infty}^{\infty}\left(e^{ist}-1-ist \right)F(s)\,ds
 \end{align*} \end{linenomath}
with
$$
F(s)=\lambda \,\frac{\cos(\alpha s)}{s\sinh(\pi s)}
+(1-2\lambda) \,\frac{1}{2s\sinh(\pi s)}\\
=\frac{2\lambda\bigl(\cos(\alpha s)-1\bigr)+1}{2s\sinh(\pi s)}.
$$
The minimum of $\cos(\alpha s)-1$ is $-2$ so that the above density $F(s)$ is non-negative
exactly when $-4\lambda+1 \geq 0$.
\end{proof}

We prove that for fixed $\nu$ the functions $g_{1,\nu,\lambda}$ increase monotonically
with $\lambda$ in the $\leqc$ order.
\begin{lemma}\label{Lemma 6.5}
If $\lambda^\prime\leq \lambda$, then $g_{1,\nu,\lambda^\prime} \leqc g_{1,\nu,\lambda}$
\end{lemma}
\begin{proof}
It suffices to show positive definiteness of 
\begin{equation}\label{(6.5)} 
\frac{f_{\nu,\lambda^\prime}(t)}{f_{\nu,\lambda}(t)}
= \frac{g_{1,\nu,\lambda^\prime}(e^t) }{g_{1,\nu,\lambda}(e^t) } 
\end{equation}
for $\lambda^\prime\leq \lambda$. However, this ratio is equal to
\begin{linenomath} \begin{align*}
\frac{1}{\cosh^{2(\lambda-\lambda^\prime)}(t/2)}
\left(\frac{\cosh t+\beta}{1+\beta}\right)^{\lambda-\lambda^\prime}
&=\left(\frac{1}{\cosh^2(t/2)} \cdot \frac{\cosh t+\beta}{1+\beta}
\right)^{\lambda-\lambda^\prime}\\
&=\left(\frac{2}{1+\beta} \cdot \frac{\cosh t+\beta}{\cosh t+1}
\right)^{\lambda-\lambda^\prime}.
 \end{align*} \end{linenomath}
On the other hand, by Lemma \ref{Lemma 5.4}\,(ii) we have
$$
\log\left(\frac{2}{1+\beta} \cdot \frac{\cosh t+\beta}{\cosh t+1}\right)
=\int_{-\infty}^{\infty}\left(e^{ist}-1-ist \right)
\frac{1-\cos(\alpha s)}{s\sinh(\pi s)}\,ds
$$
with $\alpha=\log\left(\beta+\sqrt{\beta^2-1}\right)>0$. The density
$(1-\cos(\alpha s))/s\sinh(\pi s)$ here is positive so that $(\cosh t+\beta)/(\cosh t+1)$
is infinitely divisible and hence the ratio (\ref{(6.5)}) (with
$\lambda^\prime\le\lambda$) is positive definite.
\end{proof}

The  monotonicity shown above  implies that for each fixed $\nu$, the set 
$$
\left\{ \lambda \in [0,1]:\, g_{1,\nu,\lambda} \in \cK^+\right\}
$$
is a subinterval $[0,\lambda_c]$, for which we prove that the critical value
$\lambda_c=\lambda_c(\nu)$ satisfies \eqref{geom.crit} above. The lower bound ${1\over4}$
follows immediately from Lemma \ref{Lemma 6.4}. The upper bound follows immediately from
Theorem~\ref{Theorem 6.6} below.

\begin{thm} \label{Theorem 6.6}
When $\lambda > {1\over3}$,  the function $g_{1,\nu,\lambda}$  does not belong to $\cK^+$
for any $\nu \in (0,1)$.
\end{thm}
\begin{proof}
By Theorem~\ref{thm:CP2}, this  is equivalent to showing that $f_{\nu,\lambda}(t)$ is not
positive definite when  $\lambda > {1\over3}$, for which we will prove by contradiction.  
We choose a rational $\frac{m}{n}$ with 
\begin{equation}\label{(6.6)} 
 \frac{1}{3}  < \frac{m}{n} \leq \lambda.
\end{equation}
Then by  Lemma~\ref{Lemma 6.5}, positive definiteness of $f_{\nu,\lambda}(t)$ implies
that so is $f_{\nu,\frac{m}{n}}(t)$.  Since $f_{\nu,\frac{m}{n}}(t)$ is equal to
$$
\frac{1}{\cosh^{1-\frac{2m}{n}}(t/2)\bigl(\cosh t+\beta\bigr)^{\frac{m}{n}}}
$$
up to a positive constant, its $n$th power
$$
\frac{1}{\cosh^k(t/2)\bigl(\cosh t+\beta\bigr)^m}
$$
with $k=n-2m$ would be also positive definite. However, this contradicts
Theorem \ref{Theorem 5.6}\,(i) because of (\ref{(6.6)}), i.e., $k < m$.
\end{proof}

How the critical value $\lambda_c(\nu)$ depends on $\nu \in (0,1)$ seems to be an
interesting problem. We note that $f_{\nu,1/3}(t)$ is equal to
$$
\frac{1}{\cosh^{1/3}(t/2)\left(\cosh t+\frac{1+\nu^2}{2\nu}\right)^{1/3}}
$$
up to a positive scalar. Although the function 
$\bigl(\cosh(t/2)(\cosh t+(1+\nu^2)/2\nu)\bigr)^{-1}$ is known not to be infinitely
divisible (\cite[Theorem 15]{Ko3} and see also Theorem \ref{Theorem 5.6}\,(ii)),
its cubic root might be positive definite for some  values of $\nu$). However, the authors
are unable to handle this delicate phenomenon.

\section*{Acknowledgments}
This work was begun when FH, DP and MBR were participating in the Fall, 2010 program at
the Mittag-Leffler Institute in Sweden. They are grateful for the opportunity and support
provided there. The work of MBR was partially supported by NSF grant CCF 1018401.
FH and HK acknowledge support by Grant-in-Aid for Scientific Research (C)21540208 and
(C)23540215, respectively.
We thank Ben Criger for providing the figure.

\appendix

\section{Proofs from Section~\ref{sect:prelim}}\label{app:intrep}

\subsection{Proof of Proposition~\ref{prop:lim0}} \label{sect:A1}

The following is from the proof of \cite[Corollary 2.2]{HS}. Since $k$ is an operator 
convex function on $(0,\infty)$, the function $g(x) = k(1+x)$ is operator convex on
$(-1,1)$. By Kraus' theorem \cite{Kr} (see \cite[Lemma III.1]{Ando} or
\cite[Theorem V.3.10]{Bh}) the divided difference function $h(x)\equiv(g(x) - g(0))/x$ is
operator monotone on $(-1,1)$.  Then by L\"{o}wner's integral representation
(see \cite[Theorem II.1]{Ando} or \cite[Corollary V.4.5]{Bh}) there exists a (unique)
finite measure $\mu$ on $[-1,1]$ such that
$$
h(x) = a + \int_{[-1,1]} \frac{x}{1 - \lambda x}\,d\mu(\lambda)
$$
with $a=h(0)=k'(1)$. Thus
\begin{align*}
(1+x) k(1+x)
&=(1+x)(ax+b) + (1+x)\int_{[-1,1]}\frac{x^2}{1-\lambda x}\,d\mu(\lambda) \\
&=(1+x)(ax+b) + \mu(\{-1\}) + x^2\int_{(-1,1]}\frac{1+x}{1-\lambda x}\,d\mu(\lambda)
\end{align*}
with $b=k(1)$.
Since $(x + 1)(1 - \lambda x)\leq 1$ whenever $x \in (-1,1)$ and $\lambda \in (-1,1]$, the
Lebesgue dominated convergence theorem implies that 
$$
\lim_{x \searrow -1} \int_{(-1,1]} \frac{1+x}{1 - \lambda x}\,d \mu(\lambda) = 0
$$
so that
$$
\lim_{x \searrow 0} x \,k(x) = \lim_{x \searrow -1} (1+x)\,k(1+x) = \mu(\{-1\}).
$$

\subsection{Proof of Theorem~\ref{thm:intrep}}
Since the divided difference function $h(x)\equiv (k(x)-k(1))/(x-1)$ is operator monotone
on $(0,\infty)$ as in Section \ref{sect:A1}, it is known (see, e.g.,
\cite[Theorem 1.9]{FHR})
that there exist a (unique) $\gamma\ge0$ and a (unique) positive measure $\mu$ on
$[0,\infty)$ with $\int_{[0,\infty)}(1+\lambda)^{-1}\,d\mu(\lambda)<+\infty$ such that
$$
h(x)=h(1)+\gamma(x-1)+\int_{[0,\infty)}{x-1\over x+\lambda}\,d\mu(\lambda),
\qquad x\in(0,\infty).
$$
Therefore,
\be \label{(A.1)}
k(x)=k(1)+k'(1)(x-1)+\gamma(x-1)^2+\int_{[0,\infty)}{(x-1)^2\over x+\lambda}\,d\mu(\lambda),
\quad x\in(0,\infty).
\ee
From the symmetry property $xk(x)=k\ofxinv$ we notice that
$$
\lim_{x\to\infty}{k(x)\over x}=\lim_{x\to\infty}x^{-2}k\ofxinv
=\lim_{x\searrow0}x^2k(x)
=\lim_{x\searrow0}\int_{[0,\infty)}{x^2\over x+\lambda}\,d\mu(\lambda)=0
$$
by the Lebesgue convergence theorem. On the other hand, since
$(x-1)^2/x(x+\lambda)\nearrow1$ as $x\nearrow\infty$, we have
$$
\lim_{x\to\infty}{1\over x}\int_{[0,\infty)}{(x-1)^2\over x+\lambda}\,d\mu(\lambda)
=\int_{[0,\infty)}d\mu(\lambda)
$$
by the monotone convergence theorem, and hence
$$
\lim_{x\to\infty}{k(x)\over x}=k'(1)+\gamma\cdot(+\infty)+\int_{[0,\infty)}d\mu(\lambda).
$$
Therefore, $\gamma=0$, $\mu$ is a finite measure, and
$k'(1)+\int_{[0,\infty)}d\mu(\lambda)=0$. From
$$
{(x-1)^2\over x+\lambda}=x-1-{(x-1)(1+\lambda)\over x+\lambda}
$$
it follows that
\begin{linenomath} \begin{align*}
k(x)&=k(1)-\int_{[0,\infty)}{(x-1)(1+\lambda)\over x+\lambda}\,d\mu(\lambda) \\
&=k(1)+\int_{[0,\infty)}\biggl({1+\lambda\over x+\lambda}-1\biggr)(1+\lambda)
\,d\mu(\lambda),
 \end{align*} \end{linenomath}
which is operator monotone decreasing since so is $(1+\lambda)/(x+\lambda)$. Moreover,
since $(x-1)(1+\lambda)/(x+\lambda)\nearrow1+\lambda$ as $x\nearrow\infty$, we have
$0\le k(1)-\int_{[0,\infty)}(1+\lambda)\,d\mu(\lambda)$ so that
$\int_{[0,\infty)}(1+\lambda)\,d\mu(\lambda)\le k(1)<+\infty$. Now, defining a
finite positive measure $\nu$ on $[0,\infty]$ by
$$
d\nu(\lambda)\equiv (1+\lambda)\,d\mu(\lambda)\ \ \mbox{on}\ \ [0,\infty),\quad
\nu(\{\infty\})\equiv k(1)-\int_{[0,\infty)}(1+\lambda)\,d\mu(\lambda),
$$
we write
\be \label{(A.2)}
k(x)=\int_{[0,\infty]}{1+\lambda\over x+\lambda}\,d\nu(\lambda),\qquad x\in(0,\infty),
\ee
where $(1+\lambda)/(x+\lambda)\equiv1$ for $\lambda=\infty$. Letting
$d\tilde\nu(\lambda)\equiv d\nu(\lambda^{-1})$ on $[0,\infty]$ we also write
$$
k(x)=\int_{[0,\infty]}{1+\lambda^{-1}\over x+\lambda^{-1}}\,d\tilde\nu(\lambda)
=\int_{[0,\infty]}{1+\lambda\over1+\lambda x}\,d\tilde\nu(\lambda)
$$
so that
$$
k\ofxinv=\int_{[0,\infty]}{x(1+\lambda)\over x+\lambda}\,d\tilde\nu(\lambda).
$$
This is the familiar integral expression of the operator monotone function $k\ofxinv$ with
a unique representing measure $\tilde\nu$. Hence the measure $\nu$ satisfying \eqref{(A.2)}
is unique (this fact itself is also well-known). Since
$$
k(x)=x^{-1}k\ofxinv=\int_{[0,\infty]}{1+\lambda\over1+\lambda x}\,d\nu(\lambda)
=\int_{[0,\infty]}{1+\lambda\over x+\lambda}\,\tilde\nu(\lambda),
$$
it follows that $\nu=\tilde\nu$. Define
$$
dm(\lambda)\equiv 2d\nu(\lambda)\quad \mbox{on}\quad[0,1),\qquad m(\{1\})\equiv \nu(\{1\}),
$$
to obtain
\begin{linenomath} \begin{align*}
k(x)&=\int_{[0,1)}\biggl({1+\lambda\over x+\lambda}
+{1+\lambda^{-1}\over x+\lambda^{-1}}\biggr)\,d\nu(\lambda)
+{2\over1+x}\,\nu(\{1\}) \\
&=\int_{[0,1]}{1+x\over(x+\lambda)(1+\lambda x)}\cdot{(1+\lambda)^2\over2}\,dm(\lambda).
 \end{align*} \end{linenomath}
Finally, note that the uniqueness of $m$ is immediate from that of $\nu$ in \eqref{(A.2)}.

\bigskip
The integral expression \eqref{(A.1)} was also given in \cite{LR}, which was considerably
extended in \cite[Theorem 5.1]{FHR}. There is another route to prove the two theorems in
Section \ref{sect:basics}. It was proved in \cite[Theorem 3.1]{AH} that a function
$k:(0,\infty)\to(0,\infty)$ is operator monotone decreasing if and only if it is operator
convex and non-increasing in the numerical sense. It is easy to see that if an operator
convex function $k$ satisfies the symmetry condition $xk(x)=k\ofxinv$, then it is
non-increasing numerically. Hence we have the implication (a) $\Rightarrow$ (b) of
Theorem \ref{thm:kequiv}. All other parts of Theorem \ref{thm:kequiv} are plain or
well-known. Then we can prove Theorem \ref{thm:intrep} by applying the familiar integral
expression to a symmetric operator monotone function $k\ofxinv$ as above (indeed, this
part of the proof is the same as the proof of \cite[Theorem 4.4]{KA}).

\section{Contraction bounds} \label {app:contractbd}

As stated in \eqref{contract}, monotone Riemannian metrics contract under the action of
quantum channels, i.e., CPT (CP and trace-preserving) maps.  It is well-known
\cite{HMPB, Pz1,LR, TCR} that the quasi-entropies $H_g(A,B)\equiv H_g(A,B,I)$
given by \eqref{quasi} with $K=I$ contract under CPT maps, i.e.,
$$
H_g \bigl(\Phi(A),\Phi(B)\bigr) \leq H_g(A,B),\qquad A,B\in\bP_d,
$$
whenever $g$ is operator convex on $(0,\infty)$.
In the rest of this subsection let $g(x) = (1-x)^2 k(x)$ with $k \in \cK$ as in Section
\ref{sect:back}.  In applications, the maximal contraction rate plays an important role,
which motivated in \cite{LR} the following definitions of {\em contraction coefficients}\,:
$$
\eta_k^{\rm RelEnt}(\Phi) \equiv  \sup_{\rho,\gamma \in \cD_d,\,\rho\ne \gamma}
\frac{ H_g\bigl( \Phi(\rho) ,\Phi(\gamma) \bigr) }{ H_g(\rho,\gamma) }
$$
and
$$
\eta_k^{\rm Riem}(\Phi) \equiv  \,\sup_{\rho\in\cD_d}\,\sup_{X \in \bH_d^0,\,X\ne0}
\frac{ \Gamma_{\Phi(\rho)}^k \bigl((\Phi(X), \Phi(X) \bigr) }{ \Gamma_\rho^k(X,X) }.
$$
A contraction coefficient was also defined \cite{Rusk} for the trace norm
$\norm{X}_1 \equiv \tr|X| = \tr(X^*X)^{1/2}$ distance which also contracts under CPT maps,
i.e.,  
$$
\eta^{\rm Dob}(\Phi) \equiv  \sup_{\rho,\gamma \in \cD_d,\,\rho\ne \gamma}
\frac{ \norm{ \Phi( \rho - \gamma) }_1} { \norm{\rho-\gamma}_1 },
$$
where the superscript reflects the fact that this is the quantum analogue of the classical 
Dobrushin coefficient.

For any CPT map $\Phi$, it was shown in \cite[Theorem IV.2]{LR} that 
$$
\eta_k^{\rm Riem}(\Phi) \leq \eta_k^{\rm RelEnt}(\Phi) \leq 1
\qquad \mbox{for any}\quad k \in \cK,
$$
and in \cite[Theorems 13, 14]{TKRWV} that 
\be  \label{bds}
\eta_{x^{-1/2}}^{\rm RelEnt}(\Phi) \leq \eta^{\rm Dob}(\Phi)
\leq \sqrt{\eta_k^{\rm Riem}(\Phi)}
\ee
when $k(x) $ is given by \eqref{heinz}. The upper bound in \eqref{bds},  given
in \cite[Theorem 3]{Rusk} for the particular case $k_0^{\ext}(x)=(1+x)/2x$  and in
\cite{TKRWV} for $k=\wh{k}_\alpha^\H$ in Example \ref{ex:heinz}, holds for any $k \in \cK$.  
Our work here was motivated by
the lower bound in \eqref{bds} based on the following observations from \cite{LR,TKRWV}.
Applying the max-min principle to the eigenvalue problem 
\be  \label{evalprob}
\big(\wh{\Phi} \circ \Omega^k_{\Phi(\rho)} \circ \Phi\bigr)(X)
= \lambda\,\Omega^k_\rho(X)  
\ee 
(for which $X = I$ always yields the largest eigenvalue $\lambda_1 = 1$) implies that
$$
\eta_k^{\rm Riem}(\Phi) =  \sup_{\rho \in \cD_d} \lambda_2^k(\Phi, \rho),
$$
where $\wh{\Phi}$ is the adjoint of $\Phi$ (with respect to the Hilbert-Schmidt inner
product) and $\lambda_2^k(\Phi,\rho)$ denotes the second largest eigenvalue of
\eqref{evalprob}. This is equivalent to the eigenvalue problem 
$$
\Upsilon_{\rho,\Phi}^k (\Phi(X))
= \bigl(\Omega^k_\rho\bigr)^{-1} \circ \wh{\Phi} \circ \Omega^k_{\Phi(\rho)} (\Phi(X))
= \lambda X  
$$ 
for the trace-preserving map $\Upsilon_{\rho,\Phi}^k \equiv
\bigl(\Omega^k_\rho\bigr)^{-1} \circ \wh{\Phi} \circ \Omega^k_{\Phi(\rho)}$ restricted on
$\bH_d^0$. When  $\Upsilon_{\rho,\Phi}^k$  is positivity-preserving, 
$$
\lambda_2^k(\Phi,\rho)
= \sup_{X\in\bH_d^0} \dfrac{\| \Upsilon_\rho^k (\Phi(X)) \|_1}{\|X\|_1}
\leq \sup_{X\in\bH_d^0} \dfrac{\| \Phi(X) \|_1}{\|X\|_1} = \eta^{\rm Dob}(\Phi).
$$
A sufficient condition for  $\Upsilon_{\rho,\Phi}^k$  to be positivity-preserving is that 
both $\bigl(\Omega^k_\rho\bigr)^{-1}$ and $\Omega^k_\rho$ are CP,\footnote{Unfortunately, 
in \cite{LR} it was claimed that  $\Upsilon_\rho^k$  is positivity-preserving for
$k(x) = \log x/(x-1)$.  Although $\Omega_\rho^k$ given by \eqref{BKM} is clearly
positivity-preserving, the inverse $\bigl(\Omega_\rho^k\bigr)^{-1}$ given by \eqref{BKMinv}
is not.}  which we have seen holds if and only if  $k(x) = x^{-1/2} $.  There may be
particular maps $\Phi$ for which $\Upsilon_{\rho,\Phi}^k$ is positivity-preserving
even when $\Omega^k_\rho$ and/or its inverse are not.  Whether or not the bound  
$\eta_{x^{-1/2}}^{\rm RelEnt}(\Phi) \leq \eta^{\rm Dob}(\Phi)$  holds for other $k \in\cK$
even though $\Upsilon_{\rho,\Phi}^k$ is not positivity-preserving is an open question.

\section{Some pedestrian arguments}  \label{app:ped}
  
In this section we present, for the benefit of non-experts, some very pedestrian ways to
see certain well-known results used in this paper.
  
\subsection{Functional calculus  for $L_D$ and $R_D$}  \label{app:funcalc}

It is basic that when $D$ has the spectral decomposition $D = \sum_j w_j \proj{\xi_j}$
(where $\proj{\xi_j}$ is the physicists's notation for the spectral projection onto
the eigenspace of the eigenvector $|\xi_j\>$), $\ffi(D) = \sum_j \ffi(w_j) \proj{\xi_j}$
for any function $\ffi $ on $(0,\infty)$.  It then follows that
\begin{align*}
   L_{\ffi(D)}(X)  = \ffi(L_D)(X) & = \sum_j \ffi(w_j) \proj{\xi_j} X \qquad \mbox{and} \\
   R_{\psi(D)}(X)  = \psi(R_D)(X) & = \sum_j \psi(w_j) X \proj{\xi_j}.
\end{align*}
Then the product  
$$
(\psi(R_D) \ffi(L_D))(X)
= \sum_{i,j} \ffi(w_i) \psi(w_j ) \, |\xi_i \kb \xi_j | \, \bra \xi_i, X \xi_j \ket.
$$
Since $L_D$ and $R_D$ commute, it follows that for an arbitrary function $\phi(x,y)$
$$
\phi(L_D,R_D)(X)
= \sum_{i,j} \phi(w_i, w_j) \, |\xi_i \kb \xi_j | \, \bra \xi_i, A \xi_j \ket 
$$
which is exactly the Hadamard product of  $A \circ X$ when $a_{ij} = \phi(w_i, w_j )$
and $X$ is represented in the basis $|\xi_j \ket $.

\subsection{Integral representation and inversion of BKM operator}  \label{app:BKM}

Although it is well-known that $\Omega^k_D$ and its inverse for $k(x) = \log x/(x-1)$ are
given by \eqref{BKM} and \eqref{BKMinv}, most proofs rely on an explicit expansion in 
eigenvalues  as in \cite{Lb}.     Using $L_D$  and $R_D$ allows one to see this more
directly in terms of integrals and anti-derivatives  
starting from the elementary formula 
\bee
\log x = \int_0^\infty \bigg(\frac{1}{1 + u }-\frac{1}{x + u  } \bigg)\,du
\eee
to write
\begin{linenomath} \begin{align*}
\bigl(\log L_D R_D^{-1}\bigr)(X)
& = \bigl(\log L_D - \log R_D \bigr)(X)
= \bigl( L_{\log D} - R_{\log D} \bigr)(X) \\
& = \int_0^\infty \biggl( X\,\frac{1}{D+ tI} - \frac{1}{D+ tI} X \biggr)\,dt \\ 
& = \int_0^\infty \frac{1}{D+ tI}\,( DX - XD)\,\frac{1}{D+ tI}\,dt \\
& = \int_0^\infty \frac{1}{D+ tI}\,(L_D - R_D)(X)\,\frac{1}{D+ tI}\,dt \\
\intertext{from which it follows that}
\Omega_D^k(X) & = \frac{\log L_D - \log R_D }{L_D - R_D}(X)
= \int_0^\infty \frac{1}{D+ tI}\,X\,\frac{1}{D+ tI}\,dt.
\end{align*} \end{linenomath} 
Now, observe that 
\begin{align*}
\int_0^1 D^t \bigl[ ( \log D) X - X \log D \bigr] D^{1-t}\,dt
& = \int_0^1  \frac{d~}{dt} D^t X D^{1-t}\,dt \\
& = DX - XD = (L_D - R_D)(X)
\end{align*}
so that $\int_0^1 D^t\,\Omega_D^k(X) D^{1-t}\,dt = X$, which implies \eqref{BKMinv}, i.e.,
$$
\bigl(\Omega_D^k\bigr)^{-1}(Y) = \int_0^1 D^t\,Y D^{1-t}\,dt.
$$



\end{document}